\documentclass[11pt]{article}
\usepackage[margin=1in]{geometry}
\usepackage{setspace}
\usepackage{amsmath,amssymb,amsfonts}
\usepackage{amsthm}
\usepackage{graphicx}
\usepackage{xcolor}
\usepackage{textcomp}
\usepackage{bm}
\usepackage{cases}
\usepackage{algorithm}
\usepackage{algorithmic}
\usepackage{mathtools}
\usepackage{booktabs}
\usepackage{nicefrac}
\usepackage[numbers,sort&compress]{natbib}
\usepackage{hyperref}
\hypersetup{hidelinks}
\usepackage{cleveref}

\newtheorem{theorem}{Theorem}
\newtheorem{lemma}{Lemma}
\newtheorem{corollary}{Corollary}
\newtheorem{proposition}{Proposition}

\newtheorem{definition}{Definition}
\newtheorem{assumption}{Assumption}

\renewcommand{\Pr}{\mathrm{Pr}}

\newcommand{\Br}{\mathrm{Br}}

\title{How Learning Selects Equilibria and When It Leads to Collusion}
\author{Hanzhang Wang \qquad Manxi Wu\\[6pt]
\small University of California, Berkeley}
\date{}

\begin{document}
\maketitle
\begingroup
\renewcommand{\thefootnote}{}
\footnotetext{\texttt{hz-wang@berkeley.edu}, \texttt{manxiwu@berkeley.edu}}
\endgroup

\begin{abstract}
We study how learning algorithm design shapes equilibrium selection in
infinite-horizon discounted Markov games with finite states and
actions. Each episode, players choose smooth best responses to beliefs
about opponents’ policies, test them using noisy observations, and
resample rejected beliefs or explore at rates determined by increasing
transformations of their discounted values. We show that stochastically
stable states induce approximate stationary Nash equilibria maximizing
the minimum transformed discounted value across players, with
transformations governing which players selection favors. In repeated
games with bounded recall of noisy public signals, longer recall weakly
increases the selected minimum transformed value. In repeated Cournot
games, every stochastically stable state induces an approximate public
perfect equilibrium whose minimum discounted profit across firms
exceeds the Cournot–Nash benchmark. Collusion is therefore inevitable
under these conditions and learning dynamics. For symmetric firms,
selected profits can approach the monopoly benchmark arbitrarily
closely. These results offer a novel equilibrium-selection perspective
on algorithmic collusion.
\end{abstract}


\section{Introduction}
\label{sec:introduction}
We study how the design of a learning algorithm determines equilibrium selection in infinite-horizon discounted Markov games with finite state and action spaces. When a game admits multiple equilibria, learning to play an equilibrium leaves open which outcome will emerge and persist. We ask which equilibria are selected by learning dynamics with belief testing and value-dependent exploration, and how the exploration rule can be designed to change those outcomes. We characterize the approximate stationary Nash equilibria selected in the long run and identify how players’ exploration rules determine selection. We then apply this characterization to repeated games with bounded public recall, represented as finite-state Markov games, and study its implications for algorithmic collusion in repeated Cournot competition. Under suitable conditions on patience, recall, monitoring, and learning parameters, every stochastically stable outcome exhibits collusion; for symmetric firms, selected profits can approach the monopoly benchmark arbitrarily closely. These results provide an equilibrium-selection perspective on algorithmic collusion, connecting the persistence of collusive outcomes to the design of individual learning rules.

We introduce an episodic learning dynamic in which each player maintains state-dependent beliefs about opponents’ stationary policies and chooses an entropy-regularized best response. Players hold their policies fixed during each epoch and observe opponents’ actions or informative, possibly noisy signals of those actions. At the end of the epoch, players use these observations to test whether their beliefs are consistent with opponents’ behavior. Rejected beliefs are resampled. Beliefs that pass the test may also be resampled through exploration, whose probability decreases with a player-specific increasing transformation of the discounted value evaluated under those beliefs. Conditional on passing the test, players with higher transformed values are less likely to explore, and the transformation function determines how each player’s propensity to explore responds to their discounted value. We establish conditions under which a generalized likelihood ratio test provides the error control required by the analysis.

We characterize the long-run outcomes of the learning dynamic through stochastically stable belief–policy states: those whose stationary probabilities remain positive as exploration and testing errors vanish. In this limit, players spend an arbitrarily large fraction of epochs in this set. Under the stated assumptions and suitable learning parameters, we show that these states induce approximate stationary Nash equilibria that maximize the minimum transformed discounted value across players (Theorem 1).

This result answers both questions: it establishes approximate equilibrium play in the long run and identifies how exploration design determines selection. A common increasing transformation selects outcomes maximizing the minimum value. Player-specific transformations can favor players whose exploration probabilities remain relatively high at comparable payoffs: assigning a player lower transformed values makes that player more likely to explore and can steer selection toward equilibria offering them higher payoffs. Intuitively, the player with the lowest transformed value is most likely to destabilize the current outcome, and selection therefore favors outcomes that maximizes this player’s long run utility. We formalize this intuition using Young’s stochastic-potential framework \cite{Young1993}, which identifies stochastically stable states through minimum stochastic potential.

We next apply the framework to repeated games with bounded public recall, where players condition their actions on the most recent \(L\) public observations and evaluate strategies by their discounted continuation payoffs. Viewing this history as the state yields a finite Markov game whose stationary policies correspond to strategies with bounded recall. Approximate perfect public equilibria (PPE) then correspond to approximate stationary Nash equilibria, allowing our equilibrium-selection results to apply. Holding the learning parameters fixed, we show that longer recall weakly increases the selected minimum transformed value (Theorem 2).

Applying this framework to repeated Cournot competition explains algorithmic collusion through equilibrium selection. With a common increasing transformation, learning favors consistent outcomes maximizing the lowest belief-based discounted profit across firms. Longer public recall weakly increases the selected belief-based collusion index, with a corresponding approximate monotonicity guarantee for actual profits. Under sufficient patience, recall, and monitoring accuracy, suitable learning parameters ensure that every stochastically stable outcome has a positive actual-payoff collusion index: the minimum discounted profit across firms exceeds the Cournot–Nash benchmark. For symmetric firms, this index can be made arbitrarily close to one (Theorem 3). Thus, the learning dynamics make collusion inevitable at stochastically stable outcomes under these conditions, without explicitly targeting joint-profit maximization.

\emph{Related literature.} Equilibrium selection through learning has been studied in normal-form games and stochastic games. In normal-form games, variants of log-linear learning select potential maximizers in potential games \cite{MardenShamma2012}. Completely uncoupled dynamics based on experimentation and payoff-dependent acceptance select welfare-maximizing pure Nash equilibria in generic games possessing such equilibria \cite{Pradelski2012}. Payoff-based experimentation with content/discontent states can also select aggregate-payoff-maximizing action profiles, which need not be equilibria \cite{MardenYoungPao2012}. These results show how different learning rules induce different selection objectives.

In stochastic games, Q-learning combined with policy exploration achieves team optimality in common-interest settings \cite{yongacoglu2022decentralized}, while log-linear updates combined with Q-learning yield near-efficient outcomes in stochastic teams \cite{donmez2025logitq}. Pareto Actor--Critic targets Pareto-optimal equilibria in no-conflict games \cite{christianos2023pareto}. For finite-horizon stochastic games, normal-form selection rules, including log-linear learning, are applied to continuation-value stage games, with selection properties established through backward induction \cite{zhang2024equilibrium}. Another approach uses single-bit feedback to select approximate pure stationary equilibria maximizing a sum of transformed values \cite{kiremitci2026decentralized}. Our result links individual exploration rules to max--min selection over consistent belief-policy states in infinite-horizon discounted games, without requiring a potential-game structure or pure stationary equilibria. We establish this characterization using the resistance-tree framework for stochastic stability \cite{Young1993}, applied to the process on beliefs and their induced smooth best-response policies.

Our hypothesis-testing approach builds on learning through the testing and revision of beliefs about opponents’ behavior \cite{FosterYoung2003}. In two-player repeated games, hypothesis testing combined with bounded-memory strategies and payoff-dependent experimentation yields efficient equilibrium selection and connections to bargaining solutions \cite{Jindani2022}. Our work develops this connection between experimentation and selection for discounted Markov games with multiple players, state-dependent policy beliefs, and identifiable imperfect monitoring. It also establishes how the selected minimum transformed value changes with public recall.

Our repeated-game application connects equilibrium selection to algorithmic collusion. Numerical studies show that Q-learning can sustain collusive behavior \cite{Calvano,klein2021autonomous} and examine how monitoring and algorithm design affect pricing outcomes \cite{calvano2021imperfect,asker2024impact}. Theoretical work establishes bounded-recall folk theorems \cite{MailathOlszewski2011} and characterizes learnable payoffs under smooth fictitious play in repeated potential games \cite{cartea2026algorithmic} and under Q-learning and gradient-based dynamics with finite recall \cite{AskenaziGolanEtAl2026}. We characterize how exploration selects among sustainable equilibrium outcomes, thereby determining whether competitive or collusive behavior persists. In repeated Cournot games, we establish recall comparative statics and conditions under which every stochastically stable outcome is collusive, with profits approaching the monopoly benchmark for symmetric firms. This provides an equilibrium-selection explanation of persistent algorithmic collusion and connects it directly to exploration design.

\emph{Relationship with our prior work \cite{ruifan}.} This paper builds on our prior work on equilibrium selection through episodic hypothesis testing in finite normal-form games \cite{ruifan}. Theorem 1 extends its max–min selection result to infinite-horizon discounted Markov games and allows players to observe noisy signals of opponents’ actions rather than the actions themselves. Its proof follows the same high-level structure as that in \cite{ruifan} but requires new analysis establishing (i) the sensitivity of smooth best responses to state-dependent beliefs in Lemma 4; (ii) the existence and approximate-equilibrium properties of consistent states in Proposition 1; and (iii) uniform testing-error bounds for dependent and noisy observations collected along policy-dependent state trajectories in Proposition 2. In particular, the testing analysis
accounts for dependent observations and random state-visit counts.
The repeated-game analysis, including the recall comparative statics in Theorem 2 and the algorithmic-collusion guarantees in Theorem 3, is new to this paper.

\section{Model}\label{sec: model}
Consider an infinite-horizon, discrete-time Markov game $G=(I,S,(A_i)_{i\in I},(u_i)_{i\in I}$, $P, \delta,\mu)$, where $I$ is a finite set of players with $|I|=n$; $S$ is a finite set of states; $A_i$ is the finite action set of player $i$ and $A:=\prod_{i\in I}A_i$ is the joint action set; $u_i:S\times A\to\mathbb{R}$ is player $i$'s one-stage payoff function; $P=(P(s'\mid s,a))_{s,s'\in S,\,a\in A}$ is the state transition matrix, where $P(s'\mid s,a)$ is the probability of state transitioning from $s$ to $s'$ given joint action $a$; $\delta\in(0,1)$ is the discount factor; and $\mu\in\Delta(S)$ is the initial state distribution.

For each player $i\in I$, let $\Delta_i$ denote the probability simplex over
$A_i$, and let $\Delta_{-i}:=\prod_{j\neq i}\Delta_j$. A stationary policy is a mapping $\pi_i:S\to\Delta_i$, where
$\pi_i(s,a_i)$ is the probability that player $i$ chooses action $a_i$ in
state $s$. At state $s$, the joint
mixed action is $\pi(s):=(\pi_i(s))_{i\in I}\in\Delta :=\prod_{i\in I}\Delta_i$, and
$\pi(s,a):=\prod_{i\in I}\pi_i(s,a_i)$ is the probability of joint action
$a=(a_i)_{i\in I}$. Player $i$'s expected one-stage payoff at state $s$ is
$
U_i(s,\pi)
=\mathbb{E}_{a\sim\pi(s)}[u_i(s,a)]
=\sum_{a\in A}\pi(s,a)u_i(s,a).
$

The game proceeds in discrete-time stages indexed by $t=0,1,2,\dots$. At
$t=0$, the initial state $s^0$ is sampled from $\mu$. At each stage $t$,
given state $s^t$, each player $i$ independently draws an action
$a_i^t$ according to $\pi_i(s^t)$, yielding the joint action
$a^t=(a_i^t)_{i\in I}$. The next state is drawn according to
$s^{t+1}\sim P(\cdot\mid s^t,a^t)$. Given an initial state distribution
$\mu$ and a stationary policy profile $\pi=(\pi_i)_{i\in I}$, player
$i$'s expected total discounted payoff is
$$
V_i(\mu,\pi)
:=\mathbb{E}\left[\sum_{t=0}^{\infty}\delta^t
U_i(s^t,\pi)\right],
$$
where the expectation is taken over the initial state and the state-action
trajectory induced by $\pi$ and $P$. For brevity, we refer to $V_i(\mu,\pi)$ as player $i$'s value.

\begin{definition}[$\epsilon$-stationary Nash equilibrium]
For $\epsilon\geq 0$, a stationary policy profile
$\pi^*=(\pi_i^*)_{i\in I}$ is an $\epsilon$-stationary Nash equilibrium
if, for every player $i\in I$, stationary policy $\pi_i$, and initial
state distribution $\mu\in\Delta(S)$,
\[
V_i\bigl(\mu,\pi_i^*,\pi_{-i}^*\bigr)
\geq
V_i\bigl(\mu,\pi_i,\pi_{-i}^*\bigr)-\epsilon.
\]
\end{definition}

In our learning dynamics, each player maintains state-dependent beliefs about
the opponents' policies. For computational tractability, we restrict
these beliefs to a finite grid. For a granularity parameter
$M\in\mathbb{N}_+$, let
\[
\Delta_i^M
:=
\left\{
x_i\in\Delta_i:
x_i(a_i)\in\left\{\frac{m}{M}:m=0,\ldots,M\right\},
\forall a_i\in A_i
\right\}.
\]

For each state $s\in S$ and opponent $j\neq i$, let
$b_{ij}(s)\in\Delta_j^M$ denote player $i$'s belief about player $j$'s
mixed action, where $b_{ij}(s,a_j)$ is the probability that player $i$
assigns to player $j$ choosing action $a_j$ at $s$. We assume that each
player believes the opponents' actions are conditionally independent given
the state. Thus, player $i$'s belief about all opponents at $s$ is
$b_i(s):=(b_{ij}(s))_{j\neq i}\in\prod_{j\neq i}\Delta_j^M$, and
\(b_i(s,a_{-i})
:=
\prod_{j\neq i}b_{ij}(s,a_j)\)
is the probability that player $i$ assigns to the opponents' joint action
$a_{-i}=(a_j)_{j\neq i}$. Player $i$'s belief profile across states is
$b_i:=(b_i(s))_{s\in S}\in\mathcal{B}_i$, where
$\mathcal{B}_i:=\bigl(\prod_{j\neq i}\Delta_j^M\bigr)^S$.
Let $A_{-i}:=\prod_{j\neq i}A_j$. Given $b_i(s)$, we use
\begin{align*}
U_i(s,a_i,b_i(s))
&:=
\mathbb E_{a_{-i}\sim b_i(s)}
\left[u_i(s,a_i,a_{-i})\right],\\
P(s'\mid s,a_i,b_i(s))
&:=
\mathbb E_{a_{-i}\sim b_i(s)}
\left[P(s'\mid s,a_i,a_{-i})\right]
\end{align*}
to denote player $i$'s subjective expected one-stage payoff and transition
probability, respectively. Given a stationary policy $\pi_i$ and belief
profile $b_i$, we use $V_i(s,\pi_i,b_i)$ to denote player $i$'s expected
total discounted payoff from initial state $s$ when player $i$ follows
$\pi_i$ and evaluates the opponents' actions according to $b_i$.

In our learning dynamics, player $i$ responds to their belief $b_i$ about
the opponents' stationary policies using a smooth best response (SBR).
Following the entropy-regularized Markov decision process formulation in
\cite{geist2019theory}, for $\pi_i(s)\in\Delta_i$, define the Shannon
entropy by
\(
\operatorname{Ent}(\pi_i(s))
:=
-\sum_{a_i\in A_i}\pi_i(s,a_i)\log\pi_i(s,a_i),
\)
where $0\log 0:=0$. For $\sigma>0$, the entropy-regularized value
of $\pi_i$ under belief $b_i$ is
\[
V_i^\sigma(s,\pi_i,b_i)
:=
V_i(s,\pi_i,b_i) 
+\sigma\mathbb E_{\pi_i,b_i}
\left[
\left.
\sum_{t=0}^{\infty}
\delta^t \operatorname{Ent}(\pi_i(s^t))
\,\right|\,s^0=s
\right], \quad \forall s \in S,
\]
and $V_i^\sigma(\mu,\pi_i,b_i)= \sum_{s \in S} \mu(s) V_i^\sigma(s,\pi_i,b_i)$.

\begin{definition}[Smooth best response]
For $\sigma>0$ and $b_i\in\mathcal B_i$, a stationary policy
$\Br_i^\sigma(b_i)\in(\Delta_i)^{|S|}$ is a smooth best response to $b_i$
if, for every $s\in S$,
\(
\Br_i^\sigma(b_i)
\in
\arg\max_{\pi_i\in(\Delta_i)^{|S|}}
 V_i^\sigma(s,\pi_i,b_i).
\)
\end{definition}

\begin{lemma}
\label{lemma: smooth_br}
For any $i\in I$, $\sigma>0$, and $b_i\in\mathcal B_i$, the smooth best
response $\Br_i^\sigma(b_i)$ exists and is unique. It can be computed
by regularized value iteration: given any
$V_{i,0}\in\mathbb R^{|S|}$, for $\alpha=0,1,\ldots$, define
\begin{align*}
&Q_{i,\alpha}(s,a_i)
=
U_i(s,a_i,b_i(s))
+\delta\sum_{s'\in S}
P(s'\mid s,a_i,b_i(s))V_{i,\alpha}(s'),\\
&\pi_{i,\alpha+1}(s,a_i)
=
\frac{
\exp\!\left(Q_{i,\alpha}(s,a_i)/\sigma\right)
}{
\sum_{a_i'\in A_i}
\exp\!\left(Q_{i,\alpha}(s,a_i')/\sigma\right)
},\\
&V_{i,\alpha+1}(s)
=
\sum_{a_i\in A_i}
\pi_{i,\alpha+1}(s,a_i)Q_{i,\alpha}(s,a_i)
+\sigma\operatorname{Ent}(\pi_{i,\alpha+1}(s)).
\end{align*}
Then, as $\alpha\to\infty$,
\[
V_{i,\alpha}
\to
V_i^\sigma\bigl(\cdot,\Br_i^\sigma(b_i),b_i\bigr),
\qquad
\pi_{i,\alpha+1}\to\Br_i^\sigma(b_i).
\]
\end{lemma}

Lemma~\ref{lemma: smooth_br} follows from the regularized MDP results in \cite{geist2019theory}. The detailed discussion is in Appendix~\ref{appendix:prop1}.

\section{Learning Dynamics and Stochastically Stable Set}
\label{sec:algo}

This section introduces the episodic learning dynamics and characterizes the long run outcomes of the dynamics.

Two distinct notions of state appear in this analysis. Within an epoch, the Markov game evolves across the stage states $s\in S$ introduced in Section~\ref{sec: model}. Across epochs, the learning dynamics evolve through players' belief and policy updates. To distinguish the two notions, we reserve $s$ for the state of the Markov game and refer to the pair
$z=(b,\pi)$ with $b=(b_i)_{i\in I}$ and $\pi_i=\Br_i^\sigma(b_i)$
as the \emph{belief-policy state} of the learning dynamics. The belief-policy state space is
\[
\mathcal{Z}:=\{(b,\pi)\mid b_i\in\mathcal{B}_i,\ \pi_i=\Br_i^\sigma(b_i),\ \forall i\in I\},
\]
which is finite since each $\mathcal{B}_i$ is finite.

\begin{definition}[Consistent Belief-Policy State]
\label{def:consistent_state}
For any $\tau>0$, a belief-policy state $z=(b,\pi)\in\mathcal{Z}$ is $\tau$-consistent for player $i$ if $\|b_i(s)-\pi_{-i}(s)\|_2\le\tau$ for all $s\in S$. Let
\[
\mathcal{Z}_i^\dagger:=\{z\in\mathcal{Z}\mid\|b_i(s)-\pi_{-i}(s)\|_2\le\tau,\ \forall s\in S\}
\]
denote the set of $\tau$-consistent states for player $i$. A belief-policy state $z$ is $\tau$-consistent if it is $\tau$-consistent for every player $i\in I$, and we write $\mathcal{Z}^\dagger:=\bigcap_{i\in I}\mathcal{Z}_i^\dagger$.
\end{definition}

\subsection{Learning Dynamics}
\label{sec:learning}

The learning proceeds in discrete epochs indexed by $k=1,2,\ldots$, with belief-policy state $z^k=(b^k,\pi^k)$ of epoch $k$. Each epoch consists of four steps.

\textbf{Step 1 (Play).} An initial Markov-game state $s^{k,0}$ is drawn from $\mu$. Players play the Markov game $G$ for $T(\xi)$ time steps, where $\xi \in (0, 1)$ is a hyperparameter of the learning dynamics and the epoch length $T(\xi)$ is given in Proposition~\ref{prop:GLRT}. At each time step $t=0,1,\ldots,T(\xi)-1$, the state $s^{k,t}$ is observed by all players, player $i$ chooses action $a_i^{k,t}$ according to $\pi_i^k(s^{k,t})$, the state transitions to $s^{k,t+1}\sim P(\cdot\mid s^{k,t},a^{k,t})$, and player $i$ observes a signal $o_i^{k,t}\in\mathcal{O}$ according to
$o_i^{k,t}\sim
\phi_{i,s^{k,t}}(\cdot\mid a_i^{k,t},a_{-i}^{k,t})$,
where $\mathcal{O}$ is the finite signal space. Our observation signal model covers different monitoring structures. For example, under payoff-based monitoring mechanism, each player only observes their own realized payoff, i.e., $o_i^{k,t}=u_i(s^{k,t},a_i^{k,t},a_{-i}^{k,t})$, and does not observe opponents' payoffs or actions. On the other hand, under perfect monitoring mechanism, each player observes opponents' actions, i.e., $o_i^{k,t}=a_{-i}^{k,t}$. The data collected by player $i$ during epoch $k$ is
\[
\mathcal{D}_i^k=\{(s^{k,t},a_i^{k,t},o_i^{k,t})\}_{t=0}^{T(\xi)-1}.
\]

\textbf{Step 2 (Test).} At the end of each epoch, each player $i$ independently conducts a consistency test with a fixed, player-specific probability $\gamma_i\in(0,1)$. If player $i$ conducts the test, then they use $\mathcal{D}_i^k$ to test whether their belief $b_i^k$ is $\tau$-consistent with the opponents' current policy $\pi_{-i}^k$ induced by the belief-policy state $z^k=(b^k,\pi^k)$. That is, the player forms a null hypothesis
\[
H_0:\ \|b_i^k(s)-\pi_{-i}^k(s)\|_2\le \tau,\ \text{ for all }s\in S,
\]
against the alternative hypothesis
\[
H_1:\ \text{there exists }s\in S\text{ such that }\|b_i^k(s)-\pi_{-i}^k(s)\|_2>\tau.
\]
A consistency test for player $i$ can be viewed as a function $\mathcal{C}_i:\mathcal{Z}\times\mathcal{D}_i\to\{0,1\}$, where $\mathcal{C}_i(z^k,\mathcal{D}_i^k)=1$ means the test fails to reject $H_0$ and $\mathcal{C}_i(z^k,\mathcal{D}_i^k)=0$ means the test rejects $H_0$. Type I error occurs when $z^k\in\mathcal{Z}_i^\dagger$ but $\mathcal{C}_i(z^k,\mathcal{D}_i^k)=0$, while Type II error occurs when $z^k\notin\mathcal{Z}_i^\dagger$ but $\mathcal{C}_i(z^k,\mathcal{D}_i^k)=1$.

\textbf{Step 3 (Belief Update).} Player $i$'s belief is updated according to the outcome of Step 2:

If player $i$ tests and rejects $H_0$ (i.e., $\mathcal{C}_i(z^k,\mathcal{D}_i^k)=0$), then player $i$ resamples a new belief, $b_i^{k+1}\sim\psi_i(\cdot\mid b_i^k)$, where $\psi_i=(\psi_i(b_i'\mid b_i))_{b_i,b_i'\in\mathcal{B}_i}$ is player $i$'s resampled belief distribution and $\psi_i(b_i'\mid b_i)>0$ for all $b_i, b_i' \in \mathcal{B}_i$ and all $i \in I$.

If player $i$ tests and fails to reject $H_0$ (i.e., $\mathcal{C}_i(z^k,\mathcal{D}_i^k)=1$), then player $i$ still resamples a new belief with probability $\xi^{f_i(V_i(\mu,\pi_i^k,b_i^k))}$, where $f_i$ is strictly positive and increasing. Here, we refer to the resampling as an exploration, meaning that a player still resamples even if the hypothesis testing shows that the belief is consistent with the observation. We note that players with higher value $f_i(V_i(\mu,\pi_i^k,b_i^k))$ explore with lower probability since $\xi<1$ and are therefore less likely to update  their current belief when it is consistent with observation. When player $i$ explores, a new belief is again drawn as $b_i^{k+1}\sim\psi_i(\cdot\mid b_i^k)$. If player $i$ does not explore, then their belief remains unchanged, $b_i^{k+1}=b_i^k$.

If player $i$ does not test in epoch $k$, then their belief remains unchanged, $b_i^{k+1}=b_i^k$.

\textbf{Step 4 (Policy Update).} Each player $i$ updates their policy as the smooth best response to the updated belief, i.e., $\pi_i^{k+1}=\Br_i^\sigma(b_i^{k+1})$.

We formally present the learning dynamics in Algorithm \ref{alg:learning}. We next present all the assumptions needed for the analysis.
\begin{algorithm}[t]
\caption{Episodic Learning Dynamics}
\label{alg:learning}
\renewcommand{\algorithmicrequire}{\textbf{Input:}}
\begin{algorithmic}[1]
\REQUIRE $G$; $\xi\in(0,1)$, $\sigma,\tau>0$, $M\in\mathbb N_+$; $(\gamma_i)_{i\in I}$, $T(\xi)$, $(\mathcal C_i,f_i,\psi_i)_{i\in I}$; $b_i^1\in\mathcal B_i$, $i\in I$
\STATE \textbf{Initialization:} Set $\pi_i^1\gets\Br_i^\sigma(b_i^1)$ for all $i\in I$ and $z^1\gets(b^1,\pi^1)$.
\FOR{each epoch $k=1,2,\ldots$}
    \STATE \textbf{Step 1 (Play):} Draw $s^{k,0}\sim\mu$ and set $\mathcal D_i^k\gets\varnothing$ for all $i\in I$. For $t=0,\ldots,T(\xi)-1$, each player $i$ draws $a_i^{k,t}\sim\pi_i^k(s^{k,t})$; draw $s^{k,t+1}\sim P(\cdot\mid s^{k,t},a^{k,t})$; each player observes $o_i^{k,t}\sim\phi_{i,s^{k,t}}(\cdot\mid a_i^{k,t},a_{-i}^{k,t})$ and appends $(s^{k,t},a_i^{k,t},o_i^{k,t})$ to $\mathcal D_i^k$.
    \STATE \textbf{Step 2 (Test):} Each player $i$ independently chooses to run $\mathcal C_i(z^k,\mathcal D_i^k)$ with probability $\gamma_i$.
    \STATE \textbf{Step 3 (Belief Update):} For each player $i$: if $i$ does not test, set $b_i^{k+1}\gets b_i^k$; if $i$ tests and $\mathcal C_i=0$, draw $b_i^{k+1}\sim\psi_i(\cdot\mid b_i^k)$; if $i$ tests and $\mathcal C_i=1$, make the same draw with probability $\xi^{f_i(V_i(\mu,\pi_i^k,b_i^k))}$ and otherwise set $b_i^{k+1}\gets b_i^k$.
    \STATE \textbf{Step 4 (Policy Update):} Set $\pi_i^{k+1}\gets\Br_i^\sigma(b_i^{k+1})$ for all $i\in I$ and $z^{k+1}\gets(b^{k+1},\pi^{k+1})$.
\ENDFOR
\end{algorithmic}
\end{algorithm}

\begin{assumption}[Smoothness] \label{assump: sigma, tau, M}
There exist positive constants $K_V,K_U,K_P>0$ such that for all $i\in I$, $s\in S$, $a_i\in A_i$, all stationary policies $\pi_i\in(\Delta_i)^{|S|}$, and all opponents' stationary-policy profiles $\pi_{-i},\pi_{-i}'\in(\Delta_{-i})^{|S|}$, we have
\[
\begin{aligned}
&|V_i(s,\pi_i,\pi_{-i})-V_i(s,\pi_i,\pi_{-i}')| \le K_V\|\pi_{-i}-\pi_{-i}'\|_\infty,\\
&|U_i(s,a_i,\pi_{-i}(s))-U_i(s,a_i,\pi_{-i}'(s))| \le K_U\|\pi_{-i}-\pi_{-i}'\|_\infty.\\
&\sum_{s'}|P(s'\mid s,a_i,\pi_{-i}(s))
-P(s'\mid s,a_i,\pi_{-i}'(s))| 
\le K_P\|\pi_{-i}-\pi_{-i}'\|_\infty.
\end{aligned}
\]
\end{assumption}

Assumption~\ref{assump: sigma, tau, M} ensures that players' values, expected payoffs, and transition probabilities are Lipchitz in the opponents' stationary policies.

\begin{assumption}\label{as:learning_parameters}
    Given any $\epsilon$, the parameters $(\sigma, \tau, M)$ satisfy
\begin{align}
    \sigma &\leq \epsilon(1-\delta)/2\log(\max_{i\in I}|A_i|),\quad
    \tau\leq \epsilon/4K_V, \label{eq: sigma_tau_choice} \\
    M &\geq |S||I|\max_{i\in I}|A_i|/\tau
    \left(1+|I|K_{\Br}/\sigma\right), \label{eq: M_choice}
\end{align}
where $U_{\max}:=\max_{i\in I,\,s\in S,\,a\in A}|u_i(s,a)|$, $\overline V_\sigma:=[U_{\max}+\sigma\log(\max_{i\in I}|A_i|)]/(1-\delta)$, $K_Q:=(K_U+\delta K_P\overline V_\sigma)/(1-\delta)$, and $K_{\Br}:=K_Q\sqrt{|S|\max_{i\in I}|A_i|}/2$.
\end{assumption}
The following proposition establishes the existence of consistent belief-policy states and shows that any consistent belief-policy state is an $\epsilon$-stationary Nash equilibrium.

\begin{proposition} \label{lemma: sigma, tau, M}
For any $\epsilon>0$, under Assumptions~\ref{assump: sigma, tau, M} and \ref{as:learning_parameters}, the consistent state set $\mathcal Z^\dagger$ is non-empty, and for any $z^\dagger=(b^\dagger,\pi^\dagger)\in\mathcal Z^\dagger$, $\pi^\dagger$ is an $\epsilon$-stationary Nash equilibrium.
\end{proposition}
The proof of Proposition~\ref{lemma: sigma, tau, M} is given in Appendix~\ref{appendix:prop1}.

\begin{assumption}[State reachability]
\label{assump:recurrence}
There exists a constant $p_S>0$ such that, for every pair $s,s'\in S$, there exist a finite integer $m \geq 1$ and a sequence of joint actions $(a^0,\ldots,a^{m-1})\in A^m$ satisfying
$\Pr\!\left(
s^m=s'\mid s^0=s,a^0,\ldots,a^{m-1}
\right)\geq p_S$.
\end{assumption}
Assumption~\ref{assump:recurrence} ensures that any state can be reached within finitely many steps with positive probability. This assumption is weaker than those commonly imposed in prior work, such as the requirement that every one-step transition have positive probability \cite{maheshwari2025independent}.%

\begin{assumption}[Identifiability]
\label{ass:identifiability}
There exists $\kappa>0$ such that, for each player $i$, each state $s\in S$, and any $\pi_{-i}(s),\pi'_{-i}(s)\in\Delta_{-i}$ with $\pi_{-i}(s)\neq \pi'_{-i}(s)$, there exists $a_i\in A_i$ such that
\[
\left\|\sum_{a_{-i}\in A_{-i}}\phi_{i,s}(\cdot\mid a_i,a_{-i}) \bigl[\pi_{-i}(s,a_{-i})-\pi'_{-i}(s,a_{-i})\bigr] \right\|_{\mathrm{TV}} \ge \kappa\|\pi_{-i}(s)-\pi'_{-i}(s)\|_1.
\]
\end{assumption}

Assumption \ref{ass:identifiability} ensures that the observation distribution $\phi$ enables agents to distinguish different policies at any given state of the Markov game.

\begin{assumption}[Belief Space Richness]\label{assump: bad strategy}
For any player $i\in I$, any opponents' beliefs $b_{-i}=(b_j)_{j\neq i}\in\prod_{j\neq i}\mathcal{B}_j$, and the induced policy $\pi_{-i}=\Br^\sigma_{-i}(b_{-i})$, there exists $\tilde b_i\in\mathcal{B}_i$ with $\tilde\pi_i=\Br^\sigma_i(\tilde b_i)$ such that
$\max_{s\in S}\|\tilde b_i(s)-\pi_{-i}(s)\|_2>\tau$, and
$\max_{s\in S}\|\tilde\pi_i(s)-b_{ji}(s)\|_2>\tau$ for all $j\neq i$.
\end{assumption}Assumption \ref{assump: bad strategy} ensures that, given any opponents' beliefs, there exists a belief $\tilde b_i$ of player $i$ such that the smooth best response of $\tilde b_i$ makes every player's belief inconsistent, so that all testing players reject the null hypothesis and resample with high probability. In Appendix~\ref{appendix:assumption6}, we show that this assumption holds for sufficiently small tolerance $\tau$, provided that each player’s smooth best-response map attains more distinct stationary policies than the number of players.

\begin{assumption}[Type I and type II Error Control]\label{ass:error_control}
For every player $i\in I$, the consistency test $\mathcal{C}_i$ satisfies
\begin{align}
\Pr\bigl[\mathcal{C}_i(z,\mathcal{D}_i)=0\mid z\bigr]&\le\xi^{\bar{r}},\quad \forall z\in\mathcal{Z}_i^\dagger,\label{eq:type1}\\
\Pr\bigl[\mathcal{C}_i(z,\mathcal{D}_i)=1\mid z\bigr]&\le\xi^{\bar{r}},\quad \forall z\notin\mathcal{Z}_i^\dagger,\label{eq:type2}
\end{align}
where $\bar{r} \geq \max_{z=(b,\pi)\in\mathcal{Z}}\sum_{i\in I}f_i(V_i(\mu,\pi_i,b_i))+1$ is a sufficiently large number.
\end{assumption}

Given any $\xi\in(0,1)$, the learning dynamics can adopt any testing
mechanism that satisfies Assumption~\ref{ass:error_control} with
a finite epoch length $T(\xi)$.

For concreteness, we construct
a generalized likelihood ratio test (GLRT). For any opponents' mixed-action profile
$\pi_{-i}'(s)\in\Delta_{-i}$, define the log-likelihood
of the observations at state $s$ in epoch $k$ by
\[
\ell_{i,s}^k(\pi_{-i}'(s))
:= \sum_{t:\,s^{k,t}=s}\log\left(
\sum_{a_{-i}\in A_{-i}}\pi_{-i}'(s,a_{-i})\cdot\phi_{i,s}(o_i^{k,t}\mid a_i^{k,t},a_{-i})
\right).
\]
The statewise GLRT statistic compares the unrestricted maximum
log-likelihood with the maximum over profiles consistent with
player $i$'s belief:
\begin{equation}\label{eq:glrt_statistic}
\Lambda_i^k(s)
:=
\max_{\pi_{-i}'(s)\in\Delta_{-i}}
\ell_{i,s}^k(\pi_{-i}'(s)) -
\max_{\substack{\pi_{-i}'(s)\in\Delta_{-i}\\
\|\pi_{-i}'(s)-b_i^k(s)\|_2\le\tau}}
\ell_{i,s}^k(\pi_{-i}'(s)).
\end{equation}
We use $\log 0=-\infty$ and take empty sums to be zero.
Thus, $\Lambda_i^k(s)=0$ if $s$ is not visited; if every
null candidate assigns zero likelihood to the observations,
set $\Lambda_i^k(s)=+\infty$.

For $T=T(\xi)$, the rejection threshold is given by
\begin{equation}\label{eq:glrt-threshold}
c_i(T,\xi)
:=|A_i||\mathcal O|\log(T+1)
+\log|S|+\bar r\log(1/\xi).
\end{equation}
Player $i$ rejects $H_0$ if
$\Lambda_i^k(s)>c_i(T,\xi)$ at any state, and otherwise
fails to reject it. Equivalently,
\begin{equation}\label{eq:glrt_decision}
\mathcal C_i(b_i^k,\mathcal D_i^k)
:=
\mathbf 1\{\Lambda_i^k(s)\le c_i(T,\xi),\ \forall s\in S\}.
\end{equation}
The test can be conducted using only player $i$'s belief and observations, and does not require the knowledge of opponents' utilities,
beliefs, or best-response functions.

\begin{proposition}\label{prop:GLRT}
Suppose Assumptions~\ref{assump:recurrence}
and~\ref{ass:identifiability} hold.
For every $\xi\in(0,1)$, there exists a finite
$\bar T(\xi)$ such that the GLRT test according to
\eqref{eq:glrt_statistic}--\eqref{eq:glrt_decision}
satisfies Assumption~\ref{ass:error_control}
with $T(\xi)\ge\bar T(\xi)$.
\end{proposition}

The explicit bound on $\bar T(\xi)$ and the proof are given
in Appendix~\ref{appendix:prop2}.

\subsection{Stochastically Stable Set Characterization}
\label{sec:stable}

We are now ready to characterize the long-run outcomes induced by the learning dynamics. Since each epoch restarts the Markov game from $\mu$ and the policy $\pi^k$ is determined by the beliefs $b^k$, the distribution of $z^{k+1}$ depends on the history only through $z^k$. Hence, the sequence $\{z^k\}_{k\ge1}$ is a Markov chain on the finite belief-policy state space $\mathcal{Z}$.

For any $\xi\in(0,1)$, let $P^\xi=(P^\xi_{zz'})_{z,z'\in\mathcal{Z}}$ denote the transition matrix of the belief-policy state induced by the learning dynamics, where
\[
P^\xi_{zz'}:=\Pr(z^{k+1}=z'\mid z^k=z),
\qquad \forall z,z'\in\mathcal{Z}.
\]

The parameter $\xi$ governs the perturbations in $P^\xi$. Testing errors occur with probability at most $\xi^{\bar r}$ under Assumption \ref{ass:error_control}, and player $i$ explores with probability $\xi^{f_i(V_i(\mu,\pi_i,b_i))}$ when the test fails to reject $H_0$. Both probabilities vanish as $\xi\to0$, so $\xi$ is the perturbation parameter of $P^\xi$.

\begin{lemma}[Lemma~1 in \cite{ruifan}]\label{lemma: P_ep}
For any $\xi\in(0,1)$, the Markov process $P^\xi$ is aperiodic and irreducible on the finite state space $\mathcal{Z}$, and admits a unique stationary distribution $\mu^{\xi}$.
\end{lemma}

With the unique stationary distribution $\mu^\xi$ well defined for every $\xi>0$, we study the support set of $\mu^{\xi}$ as $\xi\to0$ to identify the states that retain positive long-run probability when testing errors and exploration become arbitrarily rare.

\begin{definition}
A belief-policy state $z\in\mathcal{Z}$ is \textit{stochastically stable} relative to the process $P^\xi$ if
$\lim_{\xi\to0}\mu_z^\xi>0$,
where $\mu_z^\xi$ is the stationary probability of $z$.
\end{definition}

When $\xi$ is small, the learning dynamics spend almost all epochs at stochastically stable belief-policy states, so these states are the long-run outcomes of the learning process.

\begin{theorem}\label{theorem:main}
Under Assumptions~\ref{assump: sigma, tau, M}--\ref{ass:error_control}, a belief-policy state $z=(b,\pi)$ is stochastically stable if and only if $z\in\mathcal{Z}^\dagger$, and
\[
\min_{i\in I}f_i(V_i(\mu,\pi_i,b_i))=\max_{z'=(b',\pi')\in\mathcal{Z}^\dagger}\min_{i\in I}f_i(V_i(\mu,\pi'_i,b'_i)).
\]
When $f_i(x)=x$ for all $i\in I$, the stochastically stable states are $\epsilon$-Nash equilibria that maximize the minimum value across players, i.e. $
\min_{i\in I}V_i(\mu,\pi_i,b_i)=\max_{z'=(b',\pi')\in\mathcal{Z}^\dagger}\min_{i\in I}V_i(\mu,\pi'_i,b'_i)$.
\end{theorem}
The proof is shown in Appendix \ref{appendix:theorem1}. Theorem~\ref{theorem:main} establishes both approximate equilibrium play and a selection criterion for the learning dynamics. Hypothesis testing prevents inconsistent belief-policy states from being stochastically stable. Under the conditions of Proposition~\ref{lemma: sigma, tau, M}, every consistent state induces an approximate stationary Nash equilibrium. Value-dependent exploration then determines which of these consistent states persist in the long run.

At a consistent state, the player with the smallest transformed value has the highest exploration probability and is therefore most likely to initiate a departure. The learning dynamics select consistent states that maximize this minimum transformed value, making the largest exploration probability across players as small as possible. This characterization also reveals how exploration design shapes selection: a common strictly increasing transformation yields max--min selection on values over the consistent set, while player-specific transformations can favor different equilibrium value distributions. In particular, if one player's transformed value is the lowest at every consistent state, the selected states maximize that player's value over the consistent set.

\section{Equilibrium Selection in Repeated Games and Algorithmic Collusion}
\label{sec:repeated}

This section specializes the Markov-game framework of Section~\ref{sec:algo} to infinitely repeated games with bounded public recall. Section~\ref{sec:induced_markov_games} represents such a repeated game as a finite-state Markov game and characterizes the effect of recall length on equilibrium selection. Section~\ref{sec:collusion} applies the results to repeated Cournot competition to show that algorithmic collusion can be analyzed through the lens of equilibrium selection.

\subsection{Repeated Games with Public Monitoring}
\label{sec:induced_markov_games}

We consider an infinitely repeated game $\mathcal G$. At each stage $t$, every player $i\in I$ chooses an action $a_i^t$ from a finite set $A_i$, receives the stage payoff $u_i(a^t)$, where $a^t\in A:=\prod_{i\in I}A_i$, and observes the public signal $o^t\sim\phi(\cdot\mid a^t)$ in the finite signal space $\mathcal O$. Fix $L\geq1$. At the beginning of stage $t$, the public history $h_L^t$ records the most recent $L$ public signals, i.e., $h_L^t=(o^{t-L},\ldots,o^{t-1})$. When $t<L$, missing signals are denoted by $\varnothing\notin\mathcal O$; in particular, $h_L^0:=(\varnothing,\ldots,\varnothing)$. The history space is
$H_L:=\bigcup_{r=0}^{L}\left(\{\varnothing\}^{L-r}\times\mathcal O^r\right)$.

With a slight abuse of notation, an $L$-recall public policy of player $i$ is denoted by $\pi_{i,L}:H_L\to\Delta_i$, and the joint policy is $\pi_L=(\pi_{i,L})_{i\in I}$. Player $i$'s continuation payoff under $\pi_L$ with history $h_L^t$ is
\[
V_i(h_L^t,\pi_L):=\mathbb E_{\pi_L}\!\left[\left.\sum_{\ell=0}^{\infty}\delta^\ell u_i(a^{t+\ell})\,\right|\,h_L^t\right].
\]

\begin{definition}[$\epsilon$-perfect public equilibrium]
\label{def:epsilon-ppe}
For any $\epsilon\geq0$, a joint policy $\pi_L^*$ is an \emph{$\epsilon$-perfect public equilibrium} ($\epsilon$-PPE) of $\mathcal G$ if, for every $i\in I$, every $L$-recall public policy $\pi_{i,L}'$, and every $h_L\in H_L$,
$V_i(h_L,\pi_L^*)\geq V_i(h_L,\pi_{i,L}',\pi_{-i,L}^*)-\epsilon$.
\end{definition}

We next represent the repeated game $\mathcal G$ with $L$-recall as an induced finite-state Markov game $G_L$ with state space $H_L$. The action sets, stage payoffs, and discount factor coincide with those of $\mathcal G$, and the initial state is $h_L^0$. Let $\mu_L$ denote the initial distribution concentrated at $h_L^0$. Let $h_L=(h_{L,1},\ldots,h_{L,L})$. The transition kernel is
\[
P(h_L'\mid h_L,a):=\sum_{o\in\mathcal O}\mathbf 1\{h_L'=(h_{L,2},\ldots,h_{L,L},o)\}\,\phi(o\mid a).
\]
In $G_L$, a stationary policy maps $H_L$ to $\Delta_i$; hence, $\pi_{i,L}$ is also a stationary policy.

\begin{lemma}
\label{lem:equiv}
For any $\epsilon\geq0$, a joint policy $\pi_L^*$ is an $\epsilon$-PPE of $\mathcal G$ if and only if it is an $\epsilon$-stationary Nash equilibrium of $G_L$.
\end{lemma}

\begin{proof}
Suppose first that $\pi_L^*$ is an $\epsilon$-PPE. For every $i\in I$, $\mu\in\Delta(H_L)$, and stationary policy $\pi_{i,L}'$,
$V_i(\mu,\pi_{i,L}^*,\pi_{-i,L}^*)=\sum_{h_L\in H_L}\mu(h_L)V_i(h_L,\pi_L^*)\geq\sum_{h_L\in H_L}\mu(h_L)
V_i(h_L,\pi_{i,L}',\pi_{-i,L}^*)-\epsilon=V_i(\mu,\pi_{i,L}',\pi_{-i,L}^*)-\epsilon$.
Thus, $\pi_L^*$ is an $\epsilon$-stationary Nash equilibrium of $G_L$. Conversely, suppose that $\pi_L^*$ is an $\epsilon$-stationary Nash equilibrium. Fix $i\in I$, $h_L\in H_L$, and an $L$-recall public policy $\pi_{i,L}'$, and take $\mu\in\Delta(H_L)$ such that $\mu(h_L)=1$. Then,
$V_i(h_L,\pi_L^*)=V_i(\mu,\pi_L^*)\geq V_i(\mu,\pi_{i,L}',\pi_{-i,L}^*)-\epsilon=V_i(h_L,\pi_{i,L}',\pi_{-i,L}^*)-\epsilon$.
Thus, $\pi_L^*$ is an $\epsilon$-PPE.
\end{proof}

We next apply the learning dynamics to $G_L$. States containing
$\varnothing$ occur only during the first $L$ stages and are
therefore transient. Increasing the epoch length alone cannot
provide repeated observations at these states to ensure the existence of hypothesis testing that satisfies
Assumption~\ref{ass:error_control}. We therefore assume access
to repeated runs the game, each starting from
$h_L^0=(\varnothing,\ldots,\varnothing)$, with beliefs and
policies held fixed throughout the epoch. Observations are
pooled across runs for testing policies associated with states in the first $L$ stages. These restarts affect only
data collection; best responses and discounted values remain
defined by the original infinite-horizon game.

Every history in
$H_L$ with $r$ nonempty entries, $0\le r< L$, can be
reached at stage $r$ of a run with positive probability under
the full-support policies. Finiteness of $H_L$ and
$\mathcal Z_L$ gives a uniform positive lower bound on these
probabilities. For the GLRT in Proposition~\ref{prop:GLRT},
the conditional likelihood calculations in
Appendix~\ref{appendix:prop2} remain unchanged, and the
visit-count argument applies with repeated runs for states in the first $L$ stages.
Thus, under Assumption~\ref{ass:identifiability}, sufficiently
many runs ensure both error bounds in
Assumption~\ref{ass:error_control} for every $\xi\in(0,1)$.
Restarts therefore replace the need of
Assumption~\ref{assump:recurrence} for states in the first $L$ stages of the repeated game in establishing error
control
in the proof of Theorem~\ref{theorem:main}. Therefore, the stochastic stable set
characterization continues to apply. Let $\mathcal Z_L^\dagger$
denote the consistent belief-policy states of $G_L$, as in
Definition~\ref{def:consistent_state}.

\begin{corollary}
\label{cor:repeated-selection}
For any $L\ge1$ and $\epsilon>0$,
$\mathcal Z_L^\dagger$ is nonempty, and every
$z_L\in\mathcal Z_L^\dagger$ induces an $\epsilon$-PPE
of $\mathcal G$. Additionally,
a belief-policy state $z_L=(b_L,\pi_L)\in\mathcal Z_L$
is stochastically stable if and only if
$z_L\in\mathcal Z_L^\dagger$ and
$\min_{i\in I}f_i\!\left(
V_i(\mu_L,\pi_{i,L},b_{i,L})\right)
=
\max_{z_L'=(b_L',\pi_L')\in\mathcal Z_L^\dagger}
\min_{i\in I}f_i\!\left(
V_i(\mu_L,\pi_{i,L}',b_{i,L}')\right)$.
\end{corollary}

Nonemptiness and the equilibrium claim follow from
Proposition~\ref{lemma: sigma, tau, M}
and Lemma~\ref{lem:equiv}.
The selection characterization follows from the proof of
Theorem~\ref{theorem:main}.

\begin{theorem}
\label{theorem:monotonicity}
For any $L\ge1$ and $\delta\in(0,1)$, fix the same learning parameters $(\sigma,\tau,M)$
and exploration functions $(f_i)_{i\in I}$ in
$G_L$ and $G_{L+1}$.
Any stochastically stable states
$z_L^\dagger=(b_L^\dagger,\pi_L^\dagger)$ and
$z_{L+1}^\dagger=(b_{L+1}^\dagger,\pi_{L+1}^\dagger)$
satisfy
$\min_{i\in I}f_i\!\left(
V_i(\mu_L,\pi_{i,L}^\dagger,b_{i,L}^\dagger)\right)
\le
\min_{i\in I}f_i\!\left(
V_i(\mu_{L+1},\pi_{i,L+1}^\dagger,b_{i,L+1}^\dagger)\right)$.
\end{theorem}

\begin{proof}
Fix $z_L=(b_L,\pi_L)\in\mathcal Z_L^\dagger$. For every $h_{L+1}\in H_{L+1}$, let $h_L':=(h_{L+1,2},\ldots,h_{L+1,L+1})\in H_L$ and define
$\widetilde\pi_{i,L+1}(h_{L+1})
:=\pi_{i,L}(h_L')$, $\widetilde b_{i,L+1}(h_{L+1})
:=b_{i,L}(h_L')$ for all $i\in I$.
Since the same discretization parameter $M$ is used in both games, $\widetilde b_{i,L+1}$ belongs to the belief space of $G_{L+1}$. Moreover, player $i$'s decision problem under $\widetilde b_{i,L+1}$ at $h_{L+1}$ is identical to that under $b_{i,L}$ at $h_L'$. Thus, for every $i\in I$ and $h_{L+1}\in H_{L+1}$,
$\widetilde\pi_{i,L+1}(h_{L+1})=\pi_{i,L}(h_L')=\Br_i^\sigma(b_{i,L};h_L')
=\Br_i^\sigma(\widetilde b_{i,L+1};h_{L+1})$,
and $\|\widetilde b_{i,L+1}(h_{L+1})
-\widetilde\pi_{-i,L+1}(h_{L+1})\|_2 =\|b_{i,L}(h_L')-\pi_{-i,L}(h_L')\|_2\leq\tau$.
Therefore, $\widetilde z_{L+1}:=(\widetilde b_{L+1},\widetilde\pi_{L+1})\in\mathcal Z_{L+1}^\dagger$. Let $\widetilde{\mathcal Z}_{L+1}^\dagger:=\{\widetilde z_{L+1}:z_L\in\mathcal Z_L^\dagger\}\subseteq\mathcal Z_{L+1}^\dagger$ denote the image of this construction. Starting from $h_L^0$ and $h_{L+1}^0$, the transition kernel and the definitions above imply $z_L$ and $\widetilde z_{L+1}$ generate the same action distributions, public-signal distributions, and stage payoffs. Hence,
$V_i(\mu_{L+1},\widetilde\pi_{i,L+1},\widetilde b_{i,L+1})=V_i(\mu_L,\pi_{i,L},b_{i,L})$ for all $i\in I$.
Therefore,
\begin{align*}
&\max_{z_L\in\mathcal Z_L^\dagger}\min_{i\in I}f_i\!\left(V_i(\mu_L,\pi_{i,L},b_{i,L})\right)\\
&=\max_{\widetilde z_{L+1}\in\widetilde{\mathcal Z}_{L+1}^\dagger}\min_{i\in I}f_i\!\left(V_i(\mu_{L+1},\widetilde\pi_{i,L+1},\widetilde b_{i,L+1})\right)\\
&\leq\max_{z_{L+1}\in\mathcal Z_{L+1}^\dagger}\min_{i\in I}f_i\!\left(V_i(\mu_{L+1},\pi_{i,L+1},b_{i,L+1})\right).
\end{align*}
By Corollary~\ref{cor:repeated-selection}, each maximum equals the minimum transformed value at any stochastically stable state of the corresponding game, proving the claim.
\end{proof}
Theorem~\ref{theorem:monotonicity} shows that recall length can affect long-run equilibrium selection even when the learning parameters remain fixed: longer recall weakly increases the minimum transformed value at the selected $\epsilon$-PPE outcomes.

\subsection{Algorithmic Collusion in Cournot Competition}
\label{sec:collusion}
We now apply the equilibrium-selection results to repeated Cournot competition to examine how recall affects algorithmic collusion.
For any $L\ge1$ and let $n:=|I|\ge2$. At stage $t$, firms
simultaneously choose $q_i^t\in A_i:=\{0,\omega,2\omega,\ldots,\bar q\}$,
where $\omega>0$ and $\bar q/\omega\in\mathbb N_+$.
Let $q^t=(q_i^t)_{i\in I}$ and $Q^t:=\sum_iq_i^t$.
The inverse demand function $p:[0,n\bar q]\to\mathbb R_+$ is continuous
and nonincreasing. Firm $i$ has a continuous, nondecreasing
convex cost function $c_i$, with $c_i(0)=0$, and stage profit
$u_i(q):=q_i p(\sum_{j\in I}q_j)-c_i(q_i)$.
The cost functions may differ across firms, and inverse demand
need not be linear.

Each firm knows $q_i^t$ and observes the public signal
$o^t=q^t+\varepsilon^t\in\mathcal O$,
where $\varepsilon^t$ is drawn independently across stages and
independently of $q^t$ from a commonly known finite-support,
mean-zero distribution. We say that public monitoring has
$\eta$-small noise if
\[
\Pr\!\left(\|\varepsilon^t\|_\infty<\frac{\omega}{2}\right)
\ge1-\eta.
\]

Fix a pure strict Cournot--Nash equilibrium
$\mathbf q^N=(q_i^N)_{i\in I}$ of the stage game, and write
$u_i^N:=u_i(\mathbf q^N)$.
We assume that there exists a cooperative
profile $\mathbf q^+\in A$ such that
\begin{equation}\label{eq:cournot-cooperative-improvement}
u_i^+:=u_i(\mathbf q^+)>u_i^N,\qquad \forall i\in I.
\end{equation}
The profile $\mathbf q^+$ need not maximize joint profit.
Condition~\eqref{eq:cournot-cooperative-improvement} requires
only that it strictly improve every firm's profit over the
Nash benchmark. Define the minimum Nash profit and the maximum average joint
profit by
\[
u^N:=\min_i u_i^N,
\qquad
u^C:=\frac1n\max_{q\in A}\sum_i u_i(q).
\]
Condition~\eqref{eq:cournot-cooperative-improvement} implies
$u^C>u^N$. We use these benchmarks to define the collusion index.

\begin{definition}[Collusion Index]
\label{def:cournot-collusion-index}
For a discounted payoff $v$, define its collusion index by
\[
\Gamma(v):=\frac{(1-\delta)v-u^N}{u^C-u^N}.
\]
\end{definition}

The factor $1-\delta$ expresses discounted values in units of
one-stage profit. Applied to the minimum actual value across
firms, the index is zero when Nash profile is played in every stage. A positive index means that minimum profit
exceeds this competitive Nash benchmark.
Full collusion means maximum joint profit, and the index is one.

Throughout the analysis below, firms use the restart protocol of
Section~\ref{sec:induced_markov_games} and a common strictly positive, strictly increasing
exploration function $f_i=f$, fixed across recall lengths.
For every game and learning parameter choice under consideration,
assume that Assumptions~\ref{assump: sigma, tau, M},
\ref{as:learning_parameters}, \ref{assump: bad strategy},
and~\ref{ass:error_control} hold. Therefore, by Corollary~\ref{cor:repeated-selection}, every stochastically stable belief-policy state is consistent and induces an $\epsilon$-PPE of the repeated Cournot game. In comparisons across recall lengths, the discount factor and learning parameters are common to the games being compared.

For each $L$, choose any stochastically stable state
$z_L^\dagger=(b_L^\dagger,\pi_L^\dagger)$ and define its
subjective and actual collusion indices, respectively, by
\[
\begin{aligned}
\Gamma_L
&:=\Gamma\!\left(\min_i
V_i(\mu_L,\pi_{i,L}^\dagger,b_{i,L}^\dagger)\right),\\
\widehat{\Gamma}_L
&:=\Gamma\!\left(\min_iV_i(\mu_L,\pi_L^\dagger)\right).
\end{aligned}
\]
Subjective payoffs
$V_i(\mu_L,\pi_{i,L}^\dagger,b_{i,L}^\dagger)$ may differ
from actual payoffs $V_i(\mu_L,\pi_L^\dagger)$ because
beliefs need not coincide with opponents' policies.
Consequently, unlike $\Gamma_L$, the actual-payoff
collusion index $\widehat{\Gamma}_L(\pi_L^\dagger)$
may vary across stochastically stable states. For every recall length $L$, policy $\pi_i$, and beliefs
$b_i,b_i'$, discounted values satisfy
$\left|V_i(\mu_L,\pi_i,b_i)-V_i(\mu_L,\pi_i,b_i')\right|
\le K_V\max_{h\in H_L}\|b_i(h)-b_i'(h)\|_2$,
where a valid constant, independent of $L$, is
\[
K_V:=
\max_{i\in I}
\frac{
\bigl[\max_{\mathbf q\in A}u_i(\mathbf q)
-\min_{\mathbf q\in A}u_i(\mathbf q)\bigr]
\sqrt{\sum_{j\ne i}|A_j|}
}{2(1-\delta)^2}.
\]
This bound follows from the discounted Bellman equations
and requires no additional smoothness assumption.
Applying it with $b_i=b_{i,L}^\dagger$ and
$b_i'=\pi_{-i,L}^\dagger$, consistency gives
\begin{equation}\label{eq:cournot-actual-index-error}
\left|\widehat{\Gamma}_L(\pi_L^\dagger)-\Gamma_L\right|
\le\frac{(1-\delta)K_V\tau}{u^C-u^N}.
\end{equation}
For a fixed stage game and discount factor, this difference
vanishes as $\tau\to0$, uniformly over recall lengths
and stochastically stable states.

\begin{theorem}
\label{thm:cournot-collusion}
\textup{(i) Recall monotonicity.}
For $L\ge1$, the belief-based collusion indices satisfy
$\Gamma_L\le\Gamma_{L+1}$. For any policies
$\pi_L^\dagger$ and $\pi_{L+1}^\dagger$ induced by
stochastically stable states in the respective games,
the actual-payoff collusion indices satisfy
\begin{equation}\label{eq:cournot-actual-monotonicity}
\widehat{\Gamma}_{L+1}(\pi_{L+1}^\dagger)
\ge\widehat{\Gamma}_L(\pi_L^\dagger)
-\frac{2(1-\delta)K_V\tau}{u^C-u^N}.
\end{equation}

\textup{(ii) Inevitable collusion.}
There exist $\bar L\ge1$, $\bar\delta\in(0,1)$, and
$\bar\eta>0$ such that, for any game with recall
$L\ge\bar L$, discount factor $\delta>\bar\delta$, and
$\eta$-small monitoring noise with $\eta\le\bar\eta$,
learning parameters can be chosen so that every
stochastically stable state induces a policy
$\pi_L^\dagger$ satisfying
\[
\widehat{\Gamma}_L(\pi_L^\dagger)>0.
\]
\textup{(iii) Near-full collusion with symmetric firms.} Suppose all firms have the same cost function and $\mathbf q^N$ is symmetric. Then for every $\varepsilon>0$, there exist $\bar L\ge1$, $\bar\delta\in(0,1)$, and $\bar\eta>0$ such that, whenever $L\ge\bar L$, $\delta>\bar\delta$, and monitoring has $\eta$-small noise with $\eta\le\bar\eta$, learning parameters can be chosen so that every stochastically stable state induces a policy $\pi_L^\dagger$ satisfying
\begin{equation}\label{eq:cournot-full-selection-bound}
\widehat{\Gamma}_L(\pi_L^\dagger)\ge\frac{\max_{q\in A_i}u_i(q,\ldots,q)-u^N}{u^C-u^N}-\varepsilon .
\end{equation}
In particular, if $\max_{q\in A_i}u_i(q,\ldots,q)=u^C$, that is, if some symmetric profile maximizes joint profit on the grid, then $\widehat{\Gamma}_L(\pi_L^\dagger)\ge1-\varepsilon$.
\end{theorem}

Part~\textup{(i)} shows
that the belief-based collusion index $\Gamma_L$
weakly increases with recall length, with approximate
monotonicity for the actual-payoff collusion index
$\widehat{\Gamma}_L(\pi_L^\dagger)$. This indicates that longer recall promotes collusion.
Part~\textup{(ii)} characterizes conditions under which collusion is inevitable in that every stochastically stable state has
a positive actual-payoff collusion index, ruling out
the fully competitive Nash benchmark.
In part~\textup{(iii)}, with symmetric firms, we show that the actual-payoff collusion index approaches
one, yielding joint profits arbitrarily close to the
monopoly level. These results indicate that, although playing the
competitive Nash equilibrium in every stage remains a PPE of the repeated
game, the learning dynamics favor collusive equilibria. The extent of collusion, as measured by the
belief-based collusion index,
increases with firms' recall length. With symmetric firms, selected
outcomes can yield joint profits arbitrarily close to
the monopoly level.

\begin{proof}
\emph{Part~\textup{(i)}.}
Theorem~\ref{theorem:monotonicity}, and strict monotonicity
of $f$ imply that the selected minimum subjective value $f(\min_{i\in I} V_i(\mu_L, \pi_{i, L}^{\dagger}, b^{\dagger}_{i, L}))$ is
nondecreasing in $L$. Since the function $\Gamma(\cdot)$ is increasing,
$\Gamma_L\le\Gamma_{L+1}$. Applying
\eqref{eq:cournot-actual-index-error} at both recall lengths gives
\[
\widehat{\Gamma}_{L+1}(\pi_{L+1}^\dagger) \ge\Gamma_{L+1} -\frac{(1-\delta)K_{V}\tau}{u^C-u^N}\ge\Gamma_L -\frac{(1-\delta)K_{V}\tau}{u^C-u^N} \ge\widehat{\Gamma}_L(\pi_L^\dagger) -\frac{2(1-\delta)\tau K_{V}}{u^C-u^N}.
\]

\emph{Part~\textup{(ii)}.}
We first construct a strict equilibrium with normalized
profits within any prescribed $\epsilon_u>0$ of $u_i^+$,
then approximate it by a consistent belief-policy state.
Write $\Delta u_i:=u_i^+-u_i^N>0$ and choose an integer
$\bar L\ge1$ satisfying
\begin{equation}\label{eq:cournot-punishment-length}
\bar L\Delta u_i>
\max_{q_i'\in A_i}\{u_i(q_i',\mathbf q_{-i}^+)-u_i^+\},
\qquad \forall i.
\end{equation}
Define the trigger set
\[
\mathcal D:=\left\{o\in\mathcal O:
\min\{\|o-\mathbf q^+\|_\infty,\|o-\mathbf q^N\|_\infty\}
\ge\frac{\omega}{2}\right\}.
\]
For $L\ge\bar L$, the policy $\pi$ prescribes $\mathbf q^N$
if any of the last $\bar L$ signals belongs to $\mathcal D$,
and $\mathbf q^+$ otherwise. Empty entries do not trigger
punishment. Thus play begins with cooperation, and each
trigger starts or restarts $\bar L$ punishment periods.

The profiles $\mathbf q^+$ and $\mathbf q^N$ must differ in at
least two coordinates: otherwise the sole changing firm
would profitably deviate from the Nash profile, contradicting
Nash optimality.
Suppose the prescribed profile is $\mathbf q^+$ and firm
$i$ deviates to $q_i'\ne q_i^+$. Because outputs lie on
a grid with spacing $\omega$, the deviator's output
differs from $q_i^+$ by at least $\omega$.
Moreover, $\mathbf q^+$ and $\mathbf q^N$ differ in at
least two coordinates, so some firm $j\ne i$ still
produces $q_j^+\ne q_j^N$. Its output therefore differs
from $q_j^N$ by at least $\omega$.

When $\|\varepsilon^t\|_\infty<\omega/2$, these two
differences remain detectable:
$|o_i^t-q_i^+|>\omega/2,
|o_j^t-q_j^N|>\omega/2$.
The observed signal is consequently outside the
$\omega/2$-neighborhoods of both prescribed profiles,
so it belongs to the trigger set $\mathcal D$.
The same argument applies to a deviation from
$\mathbf q^N$. Under compliance, the signal remains
within $\omega/2$ of the prescribed profile and hence
does not trigger punishment.

Let $\nu\in\{0,\ldots,\bar L\}$ be the remaining punishment
periods, with $\nu=0$ denoting cooperation, and let
$V_{i,\nu}$ be firm $i$'s continuation value under $\pi$ with $\nu$ remaining punishment periods.
Write $p_+:=\Pr(o\in\mathcal D\mid\mathbf q^+)$ and
$p_N:=\Pr(o\in\mathcal D\mid\mathbf q^N)$; both are at most
$\eta$ from the observation accuracy assumption. The Bellman equations are
\begin{align*}
V_{i,0}&=u_i^+
+\delta[(1-p_+)V_{i,0}+p_+V_{i,\bar L}],\\
V_{i,\nu}&=u_i^N
+\delta[(1-p_N)V_{i,\nu-1}+p_NV_{i,\bar L}],
\quad 1\le\nu\le\bar L.
\end{align*}
Subtracting the second equation from the first and iterating
gives
\begin{equation}\label{eq:cournot-value-difference}
V_{i,0}-V_{i,\nu}
=
\frac{\Delta u_i\sum_{t=0}^{\nu-1}[\delta(1-p_N)]^t}
{1-\delta(p_N-p_+)\sum_{t=0}^{\bar L-1}[\delta(1-p_N)]^t}.
\end{equation}
Since $p_N,p_+\in[0,\eta]$, we have $|p_N-p_+|\le\eta$.
Moreover, each term in the denominator's geometric sum
is at most one, so the denominator differs from one by
at most $\eta\bar L$. Thus, for $\eta\bar L<1/2$, the denominator
lies in $[1-\eta\bar L,1+\eta\bar L]$ and is positive.
The numerator strictly increases with $\nu$, while the
denominator is independent of $\nu$, giving
$V_{i,0}>\cdots>V_{i,\bar L}$.
Using $p_N\le\eta$ to bound the numerator from below
and bounding its $\bar L$ summands by one from above,
we obtain
\begin{equation}\label{eq:cournot-value-bounds}
\frac{\Delta u_i}{1+\eta\bar L}
\sum_{t=0}^{\bar L-1}[\delta(1-\eta)]^t
\le V_{i,0}-V_{i,\bar L}\le2\bar L\Delta u_i.
\end{equation}

At a cooperation history, the payoff advantage of compliance
over deviating once to $q_i'\ne q_i^+$ and returning to $\pi$ is

\begin{align}
&u_i^+-u_i(q_i',\mathbf q_{-i}^+)+\delta[\Pr(o\in\mathcal D\mid q_i',\mathbf q_{-i}^+)-p_+]
(V_{i,0}-V_{i,\bar L})\notag\\
&\ge u_i^+-u_i(q_i',\mathbf q_{-i}^+)
+\frac{\delta(1-2\eta)\Delta u_i}{1+\eta\bar L}
\sum_{t=0}^{\bar L-1}[\delta(1-\eta)]^t. \label{eq:cournot-cooperation-incentive}
\end{align}
By \eqref{eq:cournot-punishment-length}, this lower bound
is strictly positive at $(\delta,\eta)=(1,0)$.
Continuity and finiteness of the action sets give
$\bar\delta\in(0,1)$ and $\bar\eta>0$ for which it remains positive
for every firm and deviation whenever
$\delta>\bar\delta$ and $\eta\le\bar\eta$. At a punishment history with $\nu\ge1$ periods remaining,
the corresponding advantage is
\begin{align}
&u_i^N-u_i(q_i',\mathbf q_{-i}^N)+\delta[\Pr(o\in\mathcal D\mid q_i',\mathbf q_{-i}^N)-p_N]
(V_{i,\nu-1}-V_{i,\bar L})\notag\\
&\ge u_i^N-u_i(q_i',\mathbf q_{-i}^N)>0,
\qquad q_i'\ne q_i^N, \label{eq:cournot-punishment-incentive}
\end{align}
using $\eta<1/2$, the ordering of continuation values,
and strictness of $\mathbf q^N$.
Thus the one-shot deviation principle establishes $\pi$
as a stationary Nash equilibrium. Since $V_i(\mu_L,\pi)=V_{i,0}$, the first Bellman equation gives
$0\le u_i^+-(1-\delta)V_i(\mu_L,\pi)
=\delta p_+(V_{i,0}-V_{i,\bar L})\le2\eta\bar L\Delta u_i<\epsilon_u$,
after reducing $\bar\eta$ so that
$\bar\eta\bar L<1/2$ and
$2\bar\eta\bar L\max_i\Delta u_i<\epsilon_u$.
These choices do not depend on $L\ge\bar L$.

\emph{Approximation by a consistent belief-policy state.}
Fix $L\ge\bar L$, $\delta>\bar\delta$, and an observation signal accuracy
with $\eta\le\bar\eta$. Let $g>0$ be the minimum one-shot
deviation advantage in
\eqref{eq:cournot-cooperation-incentive} and
\eqref{eq:cournot-punishment-incentive}, over all firms,
histories, and alternative actions.
Set $b_i=\pi_{-i}$ and
$\widetilde\pi_i=\Br_i^\sigma(b_i)$.
Because $\pi$ is pure, every $b_i$ belongs to the belief grid
for every integer $M$.
The entropy bound from Appendix~\ref{appendix:prop1} gives
\begin{equation}\label{eq:cournot-regularized-value}
0\le V_i^{\sigma,*}(h,b_i)-V_i(h,\pi)
\le\frac{\sigma\log|A_i|}{1-\delta}.
\end{equation}
For the prescribed output $q_i\in\{q_i^+,q_i^N\}$,
the unregularized $Q$-value gap is at least $g$.
Regularization reduces this gap by at most
$\delta\sigma\log|A_i|/(1-\delta)$.
If this quantity is at most $g/2$, the softmax formula gives
$1-\widetilde\pi_i(h,q_i)
\le(|A_i|-1)e^{-g/(2\sigma)}$.
All firms have the same output grid. Since $b_i=\pi_{-i}$,
 $\|b_i(h)-\widetilde\pi_{-i}(h)\|_2^2
\le2\sum_{j\ne i}[1-\widetilde\pi_j(h,q_j)]^2\le2(n-1)(|A_i|-1)^2e^{-g/\sigma}$.
Also, SBR optimality and the entropy bound imply
$V_i(\mu_L,\widetilde\pi_i,b_i)
\ge V_{i,0}-\sigma\log|A_i|/(1-\delta)$.

For a given $\epsilon_v>0$, first choose $\tau>0$ satisfying
the restriction in Assumption~\ref{as:learning_parameters}
and $K_{V}\tau\le\epsilon_v/2$.
Then choose $\sigma>0$ satisfying that assumption and
\[
\begin{gathered}
\frac{\sigma\log|A_i|}{1-\delta}
\le\min\{g/(2\delta),\epsilon_v/2\},\\
\sqrt{2(n-1)}(|A_i|-1)e^{-g/(2\sigma)}\le\tau.
\end{gathered}
\]
Finally choose $M$ sufficiently large to satisfy
Assumption~\ref{as:learning_parameters}.
Then $(b,\widetilde\pi)\in\mathcal Z_L^\dagger$ and
$V_i(\mu_L,\widetilde\pi_i,b_i)\ge V_{i,0}-\epsilon_v/2$.
The standing assumptions allow us to apply the selection
characterization to these dynamics. For every stochastically stable state, the selection
characterization therefore gives $\min_iV_i(\mu_L,\pi_{i,L}^\dagger,b_{i,L}^\dagger)
\ge\min_iV_i(\mu_L,\widetilde\pi_i,b_i)\ge\min_iV_{i,0}-\epsilon_v/2.$
For actual values, consistency further gives
$\min_iV_i(\mu_L,\pi_L^\dagger)
\ge\min_iV_i(\mu_L,\pi_{i,L}^\dagger,b_{i,L}^\dagger)
-K_{V}\tau\ge\min_iV_{i,0}-\epsilon_v$.
Combining these inequalities yields the actual-profit bound
\begin{equation}\label{eq:cournot-general-selection-bound}
\widehat{\Gamma}_L(\pi_L^\dagger)
\ge\frac{\min_i u_i^+-u^N-\epsilon_u-(1-\delta)\epsilon_v}
{u^C-u^N}.
\end{equation}
Let $\gamma:=\min_i u_i^+-u^N>0$.
Choose $\epsilon_u=\gamma/4$ in the equilibrium construction.
This fixes $\bar L$, $\bar\delta$, and $\bar\eta$.
For each game satisfying these thresholds, choose
$\epsilon_v=\gamma/[4(1-\delta)]$ and learning parameters
as above. Then
\[
\widehat{\Gamma}_L(\pi_L^\dagger)
\ge\frac{\gamma}{2(u^C-u^N)}>0,
\]
which proves part~\textup{(ii)}.

\emph{Part~\textup{(iii)}.}
Use a symmetric profile that attains $\max_{q\in A_i}u_i(q,\ldots,q)$ as the
cooperative profile in the preceding construction.
Since all firms have the same cost function and $\mathbf q^N$ is symmetric, it gives every firm profit $\max_{q\in A_i}u_i(q,\ldots,q)>u^N=u_i^N$, so all steps above apply with $\min_iu_i^+=\max_{q\in A_i}u_i(q,\ldots,q)$.
For a given $\varepsilon>0$, choose
\[
\epsilon_u=\frac{\varepsilon(u^C-u^N)}{2},
\qquad
\epsilon_v=\frac{\varepsilon(u^C-u^N)}{2(1-\delta)}.
\]
The first choice determines the recall, discount, and noise
thresholds. For each admissible game, the second determines
the learning-parameter choices. Substitution into
\eqref{eq:cournot-general-selection-bound} gives
\eqref{eq:cournot-full-selection-bound}, whose right-hand side equals $1-\varepsilon$ when $\max_{q\in A_i}u_i(q,\ldots,q)=u^C$.

\end{proof}

\section{Conclusion}

This paper characterized equilibrium selection under episodic learning with belief testing and value-dependent exploration in discounted Markov games. The analysis accommodates any consistency test satisfying the required error bounds, which we establish for a generalized likelihood ratio test. The stochastically stable states maximize the minimum transformed value among consistent states and induce approximate stationary Nash equilibria under suitable learning parameters. By representing repeated games with bounded public recall as finite-state Markov games, we showed that longer recall does not decrease the selected minimum transformed value. In Cournot competition, the belief-based collusion index is therefore nondecreasing in recall, while the actual-payoff index satisfies an approximate monotonicity bound. With sufficiently long recall, patient firms, sufficiently accurate monitoring, and suitable learning parameters, every stochastically stable state yields collusion. For symmetric firms, the selected outcomes can yield joint profits arbitrarily close to the monopoly level. Thus, value-dependent exploration can select collusive outcomes without explicitly targeting joint-profit maximization. The current implementation requires model information and coordinated epochs and restarts, while the growth of the state space limits scalability, particularly as public recall increases. Developing decentralized and scalable implementations that preserve the equilibrium-selection criterion is an important direction for future work.

\appendix
\section{Proof of Proposition~\ref{lemma: sigma, tau, M}}
\label{appendix:prop1}

\paragraph{Preliminaries for Lemma~\ref{lemma: smooth_br}}
Fix $i\in I$ and $\sigma>0$. For any $b_i\in\mathcal B_i$, apply the regularized MDP framework of \cite{geist2019theory} with reward $U_i(s,a_i,b_i(s))$, transition kernel $P(s'\mid s,a_i,b_i(s))$, discount factor $\delta$, and regularizer $-\sigma\operatorname{Ent}$. The same construction extends $V_i^\sigma$ and $\Br_i^\sigma$ to any $b_i\in(\Delta_{-i})^{|S|}$, as needed below.

Define the optimal regularized value by
\[
V_i^{\sigma,*}(s,b_i)
:=\max_{\pi_i\in(\Delta_i)^{|S|}}V_i^\sigma(s,\pi_i,b_i).
\]
The associated optimal regularized $Q$-value is
\[
Q_i^{\sigma,*}(s,a_i,b_i)
:=U_i(s,a_i,b_i(s)) +\delta\sum_{s'\in S}P(s'\mid s,a_i,b_i(s))V_i^{\sigma,*}(s',b_i).
\]

Vector norms below are over $s$ for $V_i^{\sigma,*}$ and $(s,a_i)$ for $Q_i^{\sigma,*}$.

For any $v\in\mathbb R^{|S|}$, define the regularized Bellman optimality operator by
\[
\begin{aligned}
&(\mathcal T_{i,b_i}^\sigma v)(s)
:=\max_{\pi_i(s)\in\Delta_i}\{
\sum_{a_i\in A_i}\pi_i(s,a_i)U_i(s,a_i,b_i(s))\\
&+\delta\sum_{a_i\in A_i,\,s'\in S}\pi_i(s,a_i)
P(s'\mid s,a_i,b_i(s))v(s') +\sigma\operatorname{Ent}(\pi_i(s))\}.
\end{aligned}
\]
By \citet[Proposition~2(iii) and Theorem~1]{geist2019theory}, $\mathcal T_{i,b_i}^\sigma$ is a $\delta$-contraction in the supremum norm with fixed point $V_i^{\sigma,*}(\cdot,b_i)$, attained by the unique optimal policy $\Br_i^\sigma(b_i)$. Proposition~1(i) and the negative-entropy example in Section~2.2 of that paper give
\begin{equation}\label{eq:regularized_br_formula}
\Br_i^\sigma(b_i)(s,a_i)
=\frac{\exp\!\left(Q_i^{\sigma,*}(s,a_i,b_i)/\sigma\right)}
{\sum_{a_i'\in A_i}\exp\!\left(Q_i^{\sigma,*}(s,a_i',b_i)/\sigma\right)}.
\end{equation}

Since the value update in Lemma~\ref{lemma: smooth_br} is $V_{i,\alpha+1}=\mathcal T_{i,b_i}^\sigma V_{i,\alpha}$, the stated convergence follows from the contraction property and continuity of the softmax map.

\paragraph{Proof of Proposition~\ref{lemma: sigma, tau, M}.}
Using the notation and contraction property established above, we first establish Lipschitz continuity of the smooth best response with respect to beliefs.

\begin{lemma}\label{lemma: Lipschitz}
For any $i\in I$, $\sigma>0$, and $b_i,b_i'\in(\Delta_{-i})^{|S|}$,
\begin{align}\label{eq: logit lipschitz in Br}
\|\Br_i^\sigma(b_i)-\Br_i^\sigma(b_i')\|_2
\leq\frac{K_{\Br}}{\sigma}\|b_i-b_i'\|_2.
\end{align}
\end{lemma}

\begin{proof}
Since $0\leq\operatorname{Ent}(\pi_i(s))\leq\log|A_i|\leq\log(\max_{j\in I}|A_j|)$,
\[
\|V_i^{\sigma,*}(\cdot,b_i)\|_\infty
\leq\frac{U_{\max}+\sigma\log(\max_{j\in I}|A_j|)}{1-\delta}
=\overline V_\sigma.
\]
Then Assumption~\ref{assump: sigma, tau, M} gives
\begin{equation}\label{eq:regularized_operator_lipschitz}
\begin{aligned}
&\|\mathcal T_{i,b_i}^{\sigma}v-
\mathcal T_{i,b_i'}^{\sigma}v\|_\infty \leq\max_{s,a_i}|U_i(s,a_i,b_i(s))-U_i(s,a_i,b_i'(s))|\\
&+\delta\|v\|_\infty\max_{s,a_i}\!
\|P(\cdot\mid s,a_i,b_i(s))-P(\cdot\mid s,a_i,b_i'(s))\|_1\\
&\leq (K_U+\delta K_P\|v\|_\infty)
\|b_i-b_i'\|_\infty.
\end{aligned}
\end{equation}
Therefore,
\[
\begin{aligned}
&\|V_i^{\sigma,*}(\cdot,b_i)-
V_i^{\sigma,*}(\cdot,b_i')\|_\infty =\|\mathcal T_{i,b_i}^\sigma
V_i^{\sigma,*}(\cdot,b_i)-
\mathcal T_{i,b_i'}^\sigma
V_i^{\sigma,*}(\cdot,b_i')\|_\infty\\
&\leq\|\mathcal T_{i,b_i}^\sigma
V_i^{\sigma,*}(\cdot,b_i)-
\mathcal T_{i,b_i}^\sigma
V_i^{\sigma,*}(\cdot,b_i')\|_\infty
{}+\|\mathcal T_{i,b_i}^\sigma
V_i^{\sigma,*}(\cdot,b_i')-
\mathcal T_{i,b_i'}^\sigma
V_i^{\sigma,*}(\cdot,b_i')\|_\infty\\
&\leq\delta
\|V_i^{\sigma,*}(\cdot,b_i)-
V_i^{\sigma,*}(\cdot,b_i')\|_\infty {}+(K_U+\delta K_P\overline V_\sigma)
\|b_i-b_i'\|_\infty.
\end{aligned}
\]
In the last inequality, the first term follows from the $\delta$-contraction property, and the second follows from \eqref{eq:regularized_operator_lipschitz}.
Rearranging and using the definition of $K_Q$ gives
\begin{align}\label{eq:regularized_value_lipschitz}
\|V_i^{\sigma,*}(\cdot,b_i)-
V_i^{\sigma,*}(\cdot,b_i')\|_\infty
\leq K_Q\|b_i-b_i'\|_\infty.
\end{align}
By the definition of $Q_i^{\sigma,*}$, Assumption~\ref{assump: sigma, tau, M}, and \eqref{eq:regularized_value_lipschitz},
\[
\begin{aligned}
&\|Q_i^{\sigma,*}(\cdot,\cdot,b_i)-
Q_i^{\sigma,*}(\cdot,\cdot,b_i')\|_\infty\\
&\leq(K_U+\delta K_P\overline V_\sigma)
\|b_i-b_i'\|_\infty +\delta\|V_i^{\sigma,*}(\cdot,b_i)-
V_i^{\sigma,*}(\cdot,b_i')\|_\infty\\
&\leq(K_U+\delta K_P\overline V_\sigma
+\delta K_Q)
\|b_i-b_i'\|_\infty\\
&=K_Q\|b_i-b_i'\|_\infty.
\end{aligned}
\]
At each state $s$, the softmax map in \eqref{eq:regularized_br_formula} is $1/(2\sigma)$-Lipschitz with respect to the $Q$-value vector $Q_i^{\sigma,*}(s,\cdot,b_i)$ in the supremum norm. Combining this bound across states with the preceding bound on $Q_i^{\sigma,*}$ yields
\[
\|\Br_i^\sigma(b_i)-\Br_i^\sigma(b_i')\|_\infty \leq\frac{1}{2\sigma} \|Q_i^{\sigma,*}(\cdot,\cdot,b_i)-Q_i^{\sigma,*}(\cdot,\cdot,b_i')\|_\infty \leq\frac{K_Q}{2\sigma} \|b_i-b_i'\|_\infty.
\]
Converting the norms and using the definition of $K_{\Br}$ gives
\[
\|\Br_i^\sigma(b_i)-\Br_i^\sigma(b_i')\|_2
\leq\frac{K_{\Br}}{\sigma}\|b_i-b_i'\|_2.
\]

\end{proof}

We now prove Proposition~\ref{lemma: sigma, tau, M}. The map $\pi\mapsto\bigl(\Br_i^\sigma(\pi_{-i})\bigr)_{i\in I}$ maps $\prod_{i\in I}(\Delta_i)^{|S|}$ into itself, since $\Br_i^\sigma(\pi_{-i})\in(\Delta_i)^{|S|}$ for every $i$. Lemma~\ref{lemma: Lipschitz} implies that this map is continuous. As $\prod_{i\in I}(\Delta_i)^{|S|}$ is nonempty, compact, and convex, Brouwer's fixed-point theorem gives a fixed point $\pi^\sigma\in\prod_{i\in I}(\Delta_i)^{|S|}$ satisfying $\pi_i^\sigma=\Br_i^\sigma(\pi_{-i}^\sigma)$ for every $i\in I$.

For each $i$, $j\ne i$, and $s$, choose $b_{ij}^\dagger(s)\in\Delta_j^M$ by simplex rounding so that
\[
|b_{ij}^\dagger(s,a_j)-\pi_j^\sigma(s,a_j)|\leq\frac1M,
\qquad a_j\in A_j.
\]
Let $b_i^\dagger=(b_{ij}^\dagger)_{j\ne i}$ and $\pi_i^\dagger:=\Br_i^\sigma(b_i^\dagger)$. Then
\[
\|b_i^\dagger-\pi_{-i}^\sigma\|_2
\leq\frac{|S||I|\max_{j\in I}|A_j|}{M}.
\]
Using Lemma~\ref{lemma: Lipschitz} and the triangle inequality,
\begin{align*}
\|b_i^\dagger-\pi_{-i}^\dagger\|_2
&\leq\|b_i^\dagger-\pi_{-i}^\sigma\|_2
+\|\pi_{-i}^\sigma-\pi_{-i}^\dagger\|_2\\
&\leq\|b_i^\dagger-\pi_{-i}^\sigma\|_2
+\sum_{j\ne i}\|\Br_j^\sigma(\pi_{-j}^\sigma)-
\Br_j^\sigma(b_j^\dagger)\|_2\\
&\leq\frac{|S||I|\max_{j\in I}|A_j|}{M}
+\frac{K_{\Br}}{\sigma}
\sum_{j\ne i}\|\pi_{-j}^\sigma-b_j^\dagger\|_2\\
&\leq\frac{|S||I|\max_{j\in I}|A_j|}{M}
\left(1+\frac{|I|K_{\Br}}{\sigma}\right)
\leq\tau.
\end{align*}
Hence, $\|b_i^\dagger(s)-\pi_{-i}^\dagger(s)\|_2\leq\tau$ for every $s$ and $i$, so $\mathcal Z^\dagger\ne\varnothing$.

For any initial distribution $\mu$, write
\[
V_i^\sigma(\mu,\pi_i,b_i)
:=\sum_{s\in S}\mu(s)V_i^\sigma(s,\pi_i,b_i).
\]
Since $0\leq\operatorname{Ent}(\pi_i(s))\leq\log|A_i|$,
\[
V_i(\mu,\pi_i,b_i)
\leq V_i^\sigma(\mu,\pi_i,b_i)
\leq V_i(\mu,\pi_i,b_i)+\frac{\sigma\log|A_i|}{1-\delta}.
\]
For any $\pi_i'\in(\Delta_i)^{|S|}$, the optimality of $\Br_i^\sigma(b_i)$ gives
\[
V_i(\mu,\pi_i',b_i) \leq V_i^\sigma(\mu,\pi_i',b_i) \leq V_i^\sigma(\mu,\Br_i^\sigma(b_i),b_i) \leq V_i(\mu,\Br_i^\sigma(b_i),b_i) +\frac{\sigma\log|A_i|}{1-\delta}.
\]
Taking the maximum over $\pi_i'$ yields
\begin{align}\label{eq: second term on V}
\max_{\pi_i'\in(\Delta_i)^{|S|}}V_i(\mu,\pi_i',b_i)
-V_i(\mu,\Br_i^\sigma(b_i),b_i)
\leq\frac{\sigma\log|A_i|}{1-\delta}.
\end{align}

Fix $z^\dagger=(b^\dagger,\pi^\dagger)\in\mathcal Z^\dagger$ and any deviation $\pi_i'$. Since $\mathcal Z^\dagger\subseteq\mathcal Z$, we have $\pi_i^\dagger=\Br_i^\sigma(b_i^\dagger)$. Consistency implies $\|b_i^\dagger-\pi_{-i}^\dagger\|_\infty\leq\max_{s\in S}\|b_i^\dagger(s)-\pi_{-i}^\dagger(s)\|_2\leq\tau$. Assumption~\ref{assump: sigma, tau, M} and \eqref{eq: second term on V} then yield
\[
\begin{aligned}
&V_i(\mu,\pi_i',\pi_{-i}^\dagger)
-V_i(\mu,\pi_i^\dagger,\pi_{-i}^\dagger)\\
&\leq|V_i(\mu,\pi_i',\pi_{-i}^\dagger)-V_i(\mu,\pi_i',b_i^\dagger)|
+V_i(\mu,\pi_i',b_i^\dagger)-V_i(\mu,\pi_i^\dagger,b_i^\dagger)
+|V_i(\mu,\pi_i^\dagger,b_i^\dagger)-V_i(\mu,\pi_i^\dagger,\pi_{-i}^\dagger)|\\
&\leq2K_V\tau+\frac{\sigma\log(\max_{j\in I}|A_j|)}{1-\delta}
\leq\epsilon.
\end{aligned}
\]
The last inequality follows from \eqref{eq: sigma_tau_choice}, which implies $\sigma\log(\max_{j\in I}|A_j|)/(1-\delta)\leq\epsilon/2$ and $2K_V\tau\leq\epsilon/2$.
Thus, $\pi^\dagger$ is an $\epsilon$-stationary Nash equilibrium.

\section{Proof of Proposition~\ref{prop:GLRT}}
\label{appendix:prop2}

Fix a player $i$, an epoch, and its initial belief-policy state
$z=(b,\pi)\in\mathcal Z$. Suppress the epoch superscript.
The policies remain fixed over stages $t=0,\ldots,T-1$.
For a candidate $\pi_{-i}'(s)\in\Delta_{-i}$, we define
\begin{align*}%
&\rho_i(o\mid s,a_i,\pi_{-i}'(s)):=
\sum_{a_{-i}\in A_{-i}}\pi_{-i}'(s,a_{-i})
\phi_{i,s}(o\mid a_i,a_{-i}).
\end{align*}
This is the signal distribution if player $i$ chooses $a_i$ at $s$
and the opponents use $\pi_{-i}'(s)$.
Thus $\ell_{i,s}(\pi_{-i}'(s))
=\sum_{t:s^t=s}\log\rho_i(o_i^t\mid s,a_i^t,\pi_{-i}'(s))$.
The null set $
\Pi_{i,0}(s):=\{\pi_{-i}'(s)\in\Delta_{-i}:
\|\pi_{-i}'(s)-b_i(s)\|_2\le\tau\}$
is compact and nonempty, since it contains $b_i(s)$.

To specify the conditional expectations, let $\mathcal F_t:=\sigma(z,s^0,(a_i^u,o_i^u,s^{u+1})_{0\le u<t})$. This contains the initial belief-policy state and the observations
through the current state $s^t$, before the stage-$t$ action and signal
are drawn. Under the true policies,
\begin{equation}\label{eq:glrt_conditional_law}
\begin{aligned}
&\Pr(a_i^t=a_i,o_i^t=o\mid\mathcal F_t)=\pi_i(s^t,a_i)
\rho_i(o\mid s^t,a_i,\pi_{-i}(s^t)).
\end{aligned}
\end{equation}
Define the possible signals at $(s,a_i)$ by $\mathcal O_i(s,a_i):=\{o\in\mathcal O:
\phi_{i,s}(o\mid a_i,a_{-i})>0\text{ for some }a_{-i}\}$.
The true opponents' policies have full support, so
$\rho_i(o\mid s,a_i,\pi_{-i}(s))>0$ for every
$o\in\mathcal O_i(s,a_i)$. Hence the likelihood ratios below have positive
denominators almost surely.

\medskip\noindent\textit{1. Finite likelihood representatives.}
For each state $s$, the data enter the likelihood only through the
counts of action-signal pairs $\mathbf n=(n_{a_i,o})_{a_i\in A_i,o\in\mathcal O}$, where
$n_{a_i,o}:=\sum_{t=0}^{T-1}
\mathbf1\{s^t=s,a_i^t=a_i,o_i^t=o\}$.
For any count array $n$, define
\[
L_{i,s,\mathbf n}(\pi_{-i}'(s))
:=\prod_{a_i\in A_i}\prod_{o\in\mathcal O}
\rho_i(o\mid s,a_i,\pi_{-i}'(s))^{n_{a_i,o}},
\]
where $0^0=1$. For the observed array,
$L_{i,s,n}(\pi_{-i}'(s))
=\exp\{\ell_{i,s}(\pi_{-i}'(s))\}$.
Each count $n_{a_i, o}$ is in $\{0,\ldots,T\}$, so there are at most $H_i(T):=(T+1)^{|A_i||\mathcal O|}$
array $n$. For every $n\in\{0,\ldots,T\}^{A_i\times\mathcal O}$, we choose one maximizer of the likelihood function in the entire mixed strategy set $\Delta_{-i}$ and the null hypothesis set $\Pi_{i,0}(s)$.
\[
\begin{aligned}
\widehat\pi_{-i}^{\mathbf n}(s)
&\in\operatorname*{arg\,max}_{\pi_{-i}'(s)\in\Delta_{-i}}
L_{i,s,\mathbf n}(\pi_{-i}'(s)),\\
\widehat\pi_{-i,0}^{\mathbf n}(s)
&\in\operatorname*{arg\,max}_{\pi_{-i}'(s)\in\Pi_{i,0}(s)}
L_{i,s,\mathbf n}(\pi_{-i}'(s)).
\end{aligned}
\]
We define the set of maximizers for every vector $n$
\[
\begin{aligned}
\widehat\Pi_{i,s}(T)
&:=\{\widehat\pi_{-i}^{n}(s):
       n\in\{0,\ldots,T\}^{A_i\times\mathcal O}\},\\
\widehat\Pi_{i,s,0}(T)
&:=\{\widehat\pi_{-i,0}^{n}(s):
    n\in\{0,\ldots,T\}^{A_i\times\mathcal O}\}.
\end{aligned}
\]
The maximizers exist by continuity and compactness. Since each array $n$ contributes one profile to each set, we have
\begin{equation}\label{eq:glrt_representative_cardinality}
|\widehat\Pi_{i,s}(T)|\le H_i(T),
\qquad
|\widehat\Pi_{i,s,0}(T)|\le H_i(T).
\end{equation}
Each set contains a maximizer for the realized counts and lies within
its original feasible set. Therefore,
\begin{equation}\label{eq:glrt_finite_representatives}
\begin{aligned}
\sup_{\pi_{-i}'(s)\in\Delta_{-i}}\ell_{i,s}(\pi_{-i}'(s))
&=\max_{\pi_{-i}'(s)\in\widehat\Pi_{i,s}(T)}
\ell_{i,s}(\pi_{-i}'(s)),\\
\sup_{\pi_{-i}'(s)\in\Pi_{i,0}(s)}\ell_{i,s}(\pi_{-i}'(s))
&=\max_{\pi_{-i}'(s)\in\widehat\Pi_{i,s,0}(T)}
\ell_{i,s}(\pi_{-i}'(s)).
\end{aligned}
\end{equation}
If all null likelihoods vanish, both sides of the second equality
are $-\infty$.

\medskip\noindent\textit{2. Type I error.} Fix a state $s$ and a candidate $\pi_{-i}'(s)$. We define
\[
X_t:=
\begin{cases}
\displaystyle
\frac{\rho_i(o_i^t\mid s,a_i^t,\pi_{-i}'(s))}
{\rho_i(o_i^t\mid s,a_i^t,\pi_{-i}(s))},&s^t=s,\\[5pt]
1,&s^t\ne s.
\end{cases}
\]
If $\{s^t=s\}$, following
\eqref{eq:glrt_conditional_law}, we have $\mathbb E[X_t\mid\mathcal F_t]=\sum_{a_i\in A_i}\pi_i(s,a_i)
\sum_{o\in\mathcal O_i(s,a_i)}
\rho_i(o\mid s,a_i,\pi_{-i}'(s))=1$. If $\{s^t\ne s\}$, the same expectation equals one because $X_t=1$.
Therefore $\mathbb E[X_t\mid\mathcal F_t]=1$ at every stage.

Let $M_0=1$ and $M_t:=\prod_{u=0}^{t-1}X_u$.
Since $M_t$ is $\mathcal F_t$-measurable,
$\mathbb E[M_{t+1}\mid\mathcal F_t]
=M_t\mathbb E[X_t\mid\mathcal F_t]=M_t$, and $\mathbb E[M_T\mid z]
=\mathbb E[M_{T-1}\mid z]=\cdots=M_0=1.
$.
Moreover, $M_T=\exp\{\ell_{i,s}(\pi_{-i}'(s))
-\ell_{i,s}(\pi_{-i}(s))\}$. Hence
\begin{equation}\label{eq:type1_lr_mean}
\mathbb E\!\left[
e^{\ell_{i,s}(\pi_{-i}'(s))-\ell_{i,s}(\pi_{-i}(s))}
\,\middle|\,z\right]=1.
\end{equation}

Suppose $z\in\mathcal Z_i^\dagger$. The true profile belongs to
$\Pi_{i,0}(s)$ at every state, so
\[
\max_{\pi^{'}_{-i}(s) \in \Pi_{i,0}(s)} \ell_{i, s}(\pi_{-i}^{'}(s)) \geq \ell_{i, s}(\pi_{-i}(s)).
\]
Following
\eqref{eq:glrt_finite_representatives},
\[
\Lambda_i(s)\le
\max_{\pi_{-i}'(s)\in\widehat\Pi_{i,s}(T)}
\{\ell_{i,s}(\pi_{-i}'(s))-\ell_{i,s}(\pi_{-i}(s))\}.
\]
For simplicity, we write $c_i=c_i(T,\xi)$. Following Markov's inequality, for each
fixed candidate, $\Pr\{\ell_{i,s}(\pi_{-i}'(s))-\ell_{i,s}(\pi_{-i}(s))>c_i\mid z\}=\Pr\{M_T>e^{c_i}\mid z\}
\le e^{-c_i}\mathbb E[M_T\mid z]=e^{-c_i}$.
Taking a union bound over the finite representative set yields $\Pr\{\Lambda_i(s)>c_i\mid z\}\le\sum_{\pi_{-i}'(s)\in\widehat\Pi_{i,s}(T)}
\Pr\{\ell_{i,s}(\pi_{-i}'(s))-\ell_{i,s}(\pi_{-i}(s))>c_i\mid z\}\le H_i(T)e^{-c_i}$. A second union bound, now over all states, proves $
\Pr\{\mathcal C_i(b_i,\mathcal D_i)=0\mid z\}\le |S|H_i(T)e^{-c_i(T,\xi)}=\xi^{\bar r}$.
This Type I bound holds for every positive integer $T$.

\medskip\noindent\textit{3. Separation from null candidates.}
Define
\begin{equation*}
\pi_{\min}:=\min_{z=(b,\pi)\in\mathcal Z}
\min_{\substack{j\in I,\ s\in S\\a_j\in A_j}}\pi_j(s,a_j)>0.
\end{equation*}
Positivity follows from the full support of smooth best responses
and the finiteness of $\mathcal Z$. Define also
\begin{equation*}
d:=\min_{\substack{j\in I,\ z=(b,\pi)\in\mathcal Z\\
z\notin\mathcal Z_j^\dagger}}
\left(\max_{s\in S}\|b_j(s)-\pi_{-j}(s)\|_2-\tau\right)>0.
\end{equation*}
Every term in this minimum is strictly positive by inconsistency,
so finiteness of $\mathcal Z$ gives $d>0$.
If the minimizing set is empty, set $d=1$; then no Type II bound
is needed.

For $z\notin\mathcal Z_i^\dagger$, choose $s^*$ maximizing
$\|b_i(s)-\pi_{-i}(s)\|_2$, so this distance is at least $\tau+d$.
Every $\pi_{-i}'(s^*)\in\Pi_{i,0}(s^*)$ satisfies $\|\pi_{-i}(s^*)-\pi_{-i}'(s^*)\|_2\ge\|\pi_{-i}(s^*)-b_i(s^*)\|_2
-\|\pi_{-i}'(s^*)-b_i(s^*)\|_2 \ge(\tau+d)-\tau=d$.
By Assumption~\ref{ass:identifiability}, there exists an action
$a_i^*$ such that
\begin{equation}\label{eq:signal_tv_gap}
\begin{aligned}
&\bigl\|\rho_i(\cdot\mid s^*,a_i^*,\pi_{-i}(s^*)) -\rho_i(\cdot\mid s^*,a_i^*,\pi_{-i}'(s^*))\bigr\|_{\rm TV}\\
&\ge\kappa\|\pi_{-i}(s^*)-\pi_{-i}'(s^*)\|_1 \ge\kappa\|\pi_{-i}(s^*)-\pi_{-i}'(s^*)\|_2\ge\kappa d.
\end{aligned}
\end{equation}
Define the signal affinity
\begin{equation}\label{eq:glrt_affinity}
h_i(a_i;\pi_{-i}'(s^*)):= \sum_{o\in\mathcal O}
\sqrt{\rho_i(o\mid s^*,a_i,\pi_{-i}(s^*))}
\sqrt{\rho_i(o\mid s^*,a_i,\pi_{-i}'(s^*))}.
\end{equation}
For distributions $p,\widetilde p$ with affinity
$h=\sum_o\sqrt{p(o)\widetilde p(o)}$, Cauchy--Schwarz gives
$0\le h\le1$ and $
\|p-\widetilde p\|_{\rm TV}^2\le\tfrac14
\sum_o(\sqrt{p(o)}-\sqrt{\widetilde p(o)})^2\cdot
\sum_o(\sqrt{p(o)}+\sqrt{\widetilde p(o)})^2=\tfrac14(2-2h)(2+2h)\le2(1-h)$.
Consequently, \eqref{eq:signal_tv_gap} implies
$h_i(a_i^*;\pi_{-i}'(s^*))\le1-\kappa^2d^2/2$,
whereas the affinity at every other action is at most one.
With $\beta:=\pi_{\min}\kappa^2d^2/2>0$, we obtain
\begin{equation}\label{eq:uniform_hellinger_gap}
\begin{aligned}
\sum_{a_i}\pi_i(s^*,a_i)h_i(a_i;\pi_{-i}'(s^*)) &\le\pi_i(s^*,a_i^*)(1-\kappa^2d^2/2)
+\sum_{a_i\ne a_i^*}\pi_i(s^*,a_i)\\
&=1-\pi_i(s^*,a_i^*)\kappa^2d^2/2
\le1-\beta\le e^{-\beta}.
\end{aligned}
\end{equation}

\medskip\noindent\textit{4. Type II error given sufficiently many state visits.}
Fix a null candidate $\pi_{-i}'(s^*)$ and use $X_t$ defined above
with $s=s^*$. Following
\eqref{eq:glrt_conditional_law} and \eqref{eq:uniform_hellinger_gap}, if $\{s^t=s^*\}$,
\[
\mathbb E[\sqrt{X_t}\mid\mathcal F_t]
=\sum_{a_i}\pi_i(s^*,a_i)h_i(a_i;\pi_{-i}'(s^*))
\le e^{-\beta}.
\]
If $\{s^t\neq s^*\}$, $X_t=1$. Hence
$Y_t:=e^{\beta\mathbf1\{s^t=s^*\}}\sqrt{X_t}$ satisfies
$\mathbb E[Y_t\mid\mathcal F_t]\le1$ and
\[
\mathbb E\!\left[\prod_{u=0}^tY_u\mid\mathcal F_t\right]
=\left(\prod_{u=0}^{t-1}Y_u\right)
\mathbb E[Y_t\mid\mathcal F_t]
\le\prod_{u=0}^{t-1}Y_u.
\]
Iterating expectations gives
$\mathbb E[\prod_{t=0}^{T-1}Y_t\mid z]\le1$.
Writing $T_{s^*}:=\sum_{t=0}^{T-1}\mathbf1\{s^t=s^*\}$
and expanding this product yields
\begin{equation}\label{eq:type2_factor_expectation}
\mathbb E\!\left[
\exp\!\left\{
\frac{\ell_{i,s^*}(\pi_{-i}'(s^*))-\ell_{i,s^*}(\pi_{-i}(s^*))}{2}
+\beta T_{s^*}
\right\}\,\middle|\,z\right]\le1.
\end{equation}
For $N\ge0$, let $\mathcal E_N(\pi_{-i}'(s^*))$ be the event that
$\ell_{i,s^*}(\pi_{-i}'(s^*))-\ell_{i,s^*}(\pi_{-i}(s^*))\ge-c_i$
and $T_{s^*}\ge N$.
On this event, the exponential in
\eqref{eq:type2_factor_expectation} is at least
$e^{-c_i/2+\beta N}$, so Markov's inequality gives
\begin{equation}\label{eq:type2_fixed_candidate}
\Pr\{\mathcal E_N(\pi_{-i}'(s^*))\mid z\}
\le e^{c_i/2-\beta N}.
\end{equation}
If $\mathcal C_i(b_i,\mathcal D_i)=1$ and $T_{s^*}\ge N$,
the definition of the GLRT and
\eqref{eq:glrt_finite_representatives} imply that
$\mathcal E_N(\pi_{-i}'(s^*))$ occurs for some
$\pi_{-i}'(s^*)\in\widehat\Pi_{i,s^*,0}(T)$.
Following \eqref{eq:glrt_representative_cardinality} and \eqref{eq:type2_fixed_candidate}, a union bound gives
\begin{equation}\label{eq:visit_threshold}
\Pr\{\mathcal C_i(b_i,\mathcal D_i)=1,\ T_{s^*}\ge N\mid z\} \le\sum_{\pi_{-i}'(s^*)\in\widehat\Pi_{i,s^*,0}(T)} \Pr\{\mathcal E_N(\pi_{-i}'(s^*))\mid z\} \le H_i(T)e^{c_i/2-\beta N}.
\end{equation}

\medskip\noindent\textit{Number of visits.}
Let $p_S$ be as in Assumption~\ref{assump:recurrence}. With a slight abuse of notation, let $m$ denote the maximum over all pairs $s,s'\in S$ of the integers in Assumption~\ref{assump:recurrence}, which is finite since $S$ is finite. Divide the epoch into $n_T:=\lfloor T/(m+1)\rfloor$ complete
blocks of $m+1$ stages, ignoring the remainder.
Let $J_\chi$ indicate a visit to $s^*$ in block $\chi$,
whose first stage is $t_\chi=(\chi-1)(m+1)$.
From any initial state, reachability provides at most $m$
prescribed joint actions that reach $s^*$ with probability
at least $p_S$. Each prescribed joint action has probability
at least $\pi_{\min}^{|I|}$ under the actual policies.
Multiplying along each path and summing over paths reaching
$s^*$ therefore gives
\begin{equation}\label{eq:glrt_reachability_probability}
\mathbb E[J_\chi\mid\mathcal F_{t_\chi}]
=\Pr(J_\chi=1\mid\mathcal F_{t_\chi})
\ge p_S\pi_{\min}^{|I|m}.
\end{equation}
Also, $T_{s^*}\ge\sum_{\chi=1}^{n_T}J_\chi$, since distinct
blocks contribute distinct visits. For $\theta>0$, conditional Hoeffding's lemma applied to the
centered indicator gives
\[
\mathbb E[e^{-\theta J_\chi}\mid\mathcal F_{t_\chi}]\le e^{-\theta\mathbb E[J_\chi\mid\mathcal F_{t_\chi}]+\theta^2/8}\le e^{-\theta p_S\pi_{\min}^{|I|m}+\theta^2/8}.
\]
Iterating over the blocks, without assuming independence, yields
\[
\mathbb E\!\left[e^{-\theta\sum_{\chi=1}^{n_T}J_\chi}\mid z\right]\le e^{n_T(-\theta p_S\pi_{\min}^{|I|m}+\theta^2/8)}.
\]
Markov's inequality then gives
\[
\Pr\!\left\{\sum_{\chi=1}^{n_T}J_\chi<n_Tp_S\pi_{\min}^{|I|m}/2|z\right\}\le e^{\theta n_Tp_S\pi_{\min}^{|I|m}/2}\mathbb E\!\left[e^{-\theta\sum_{\chi=1}^{n_T}J_\chi}\mid z\right] \le e^{-\theta n_Tp_S\pi_{\min}^{|I|m}/2+n_T\theta^2/8}.
\]
Taking $\theta=2p_S\pi_{\min}^{|I|m}$ proves
\begin{equation}\label{eq:glrt_visit_tail}
\Pr\!\left\{T_{s^*}<\frac{n_Tp_S\pi_{\min}^{|I|m}}2
\,\middle|\,z\right\}
\le e^{-n_Tp_S^2\pi_{\min}^{2|I|m}/2}.
\end{equation}
Split the Type II event according to this visit threshold.
Using \eqref{eq:visit_threshold} with
$N=n_Tp_S\pi_{\min}^{|I|m}/2$ yields
\begin{equation}\label{eq:glrt_type2_bound}
\Pr\{\mathcal C_i(b_i,\mathcal D_i)=1\mid z\} \le e^{-n_Tp_S^2\pi_{\min}^{2|I|m}/2}
+H_i(T)e^{c_i/2-\beta n_Tp_S\pi_{\min}^{|I|m}/2}.
\end{equation}

\medskip\noindent\textit{Sufficient epoch length.}
Set $D:=|\mathcal O|\max_{j\in I}|A_j|$ and
\begin{equation}\label{eq:glrt_length_constants}
\lambda:=
\frac{p_S\pi_{\min}^{|I|m}
\min\{p_S\pi_{\min}^{|I|m},\beta\}}{4(m+1)}>0.
\end{equation}
For $T\ge2(m+1)$, $n_T\ge T/[2(m+1)]$. Thus
\[
\begin{aligned}
\frac{n_Tp_S^2\pi_{\min}^{2|I|m}}2
&\ge\frac{Tp_S^2\pi_{\min}^{2|I|m}}{4(m+1)}
\ge\lambda T,\\
\frac{\beta n_Tp_S\pi_{\min}^{|I|m}}2
&\ge\frac{T\beta p_S\pi_{\min}^{|I|m}}{4(m+1)}
\ge\lambda T.
\end{aligned}
\]
Moreover,
$H_i(T)e^{c_i/2}
=\sqrt{|S|}(T+1)^{3|A_i||\mathcal O|/2}\xi^{-\bar r/2}
\le\sqrt{|S|}(T+1)^{3D/2}\xi^{-\bar r/2}$,
and the latter is at least one. Hence
\eqref{eq:glrt_type2_bound} implies
\begin{equation}\label{eq:glrt_type2_simplified}
\Pr\{\mathcal C_i(b_i,\mathcal D_i)=1\mid z\} \le
2\sqrt{|S|}(T+1)^{3D/2}\xi^{-\bar r/2}e^{-\lambda T}.
\end{equation}
Using $\log(1+x)\le ax+\log(1+1/a)$ for $x\ge0,a>0$,
with $x=T$ and $a=\lambda/(3D)$, gives
\[
\frac{3D}{2}\log(T+1)
\le\frac{\lambda T}{2}
+\frac{3D}{2}\log(1+3D/\lambda).
\]
The logarithm of the bound in
\eqref{eq:glrt_type2_simplified} is therefore at most
\[
\log(2\sqrt{|S|})
+\frac{3D}{2}\log(1+3D/\lambda)
+\frac{\bar r}{2}\log(1/\xi)-\frac{\lambda T}{2}.
\]
Choose
\begin{equation}\label{eq:glrt_epoch_bound}
\bar T(\xi):=\Bigg\lceil\max\Bigg\{2(m+1), \frac{\log(4|S|)+3\bar r\log(1/\xi)+3D\log(1+3D/\lambda)}{\lambda} \Bigg\}\Bigg\rceil.
\end{equation}
For every integer $T\ge\bar T(\xi)$, the preceding logarithmic
bound is at most $-\bar r\log(1/\xi)$, so the Type II error
is at most $\xi^{\bar r}$. Together with the bound of type I error in step 2, this proves
Assumption~\ref{ass:error_control} holds for all $i$ and $z$.
\hfill$\square$

\section{Sufficient Condition for Assumption~\ref{assump: bad strategy}}
\label{appendix:assumption6}
\begin{lemma}\label{lemma:simple_condition}
Suppose
$\bigl|\{\Br_i^\sigma(b_i):b_i\in\mathcal B_i\}\bigr|>|I|$ for all $i\in I$.
There exists $\bar\tau>0$ such that Assumption~\ref{assump: bad strategy} holds for every $\tau\in(0,\bar\tau)$.
\end{lemma}

The proof of \citet[Lemma~7]{ruifan} applies to full stationary-policy profiles. Specifically, regard each $b_{ji}$ and $\Br_i^\sigma(b_i)$ as a point in $(\Delta_i)^{|S|}$ and replace each distance in that proof by the corresponding maximum statewise Euclidean distance. The image-size condition gives a strictly positive separation margin for every $i$ and $b_{-i}$. Since the player set and belief spaces are finite, the minimum of these margins is some $\bar\tau>0$. Hence, Assumption~\ref{assump: bad strategy} holds for every $\tau\in(0,\bar\tau)$.

\section{Proof of Theorem~\ref{theorem:main}}
\label{appendix:theorem1}

Our proof follows the three-step approach of our prior work \cite{ruifan}, adapted to state-dependent beliefs and stationary policies in discounted Markov games. First, we establish regular perturbation and compute transition resistances in terms of belief-based discounted values, using Proposition~\ref{prop:GLRT} to control testing errors along the induced state trajectories. Second, we show that the recurrent classes are singleton consistent states, whose existence and approximate stationary Nash equilibrium property follow from Proposition~\ref{lemma: sigma, tau, M}. Here, consistency requires beliefs to agree with opponents’ policies within the prescribed tolerance at every game state.
Third, we construct paths between these states and compute their minimum resistances, from which the resistance-tree criterion of \citet{Young1993} yields the max--min selection rule. %

\medskip\noindent\textit{1. Regular perturbation.} Fix all game and learning parameters except $\xi$, and choose $T(\xi)$ so that Assumption~\ref{ass:error_control} holds. For $z=(b,\pi)\in\mathcal Z$, write
$I_c(z):=\{i\in I:z\in\mathcal Z_i^\dagger\}$.
We call player $i$ consistent at $z$ if $\|b_i(s)-\pi_{-i}(s)\|_2\le\tau$ for every $s\in S$, and inconsistent otherwise. Thus $I_c(z)$ is the set of consistent players, and $z\in\mathcal Z^\dagger$ exactly when $I_c(z)=I$.

In the unperturbed process $P^0$, both Type I and Type II error probabilities are zero, and players do not explore when the null hypothesis is not rejected. A consistent player retains its belief whether or not it tests. An inconsistent player retains its belief if it does not test, which occurs with probability $1-\gamma_i$; otherwise, the test rejects and the player draws a belief according to $\psi_i(\cdot\mid b_i)$. The independent update randomizations give $P^0_{zz'}=\prod_{i\in I}P^0_{i,zz'}$, where $P^0_{i,zz'}$ is the probability that player $i$ updates its belief from $b_i$ to $b_i'$ under $P^0$, given the current state $z$. Specifically,
\begin{equation}\label{eq:blue-P0}
P^0_{i,zz'}=
\begin{cases}
\mathbf 1\{b_i'=b_i\}, & i\in I_c(z),\\
(1-\gamma_i)\mathbf 1\{b_i'=b_i\}
+\gamma_i\psi_i(b_i'\mid b_i), & i\notin I_c(z).
\end{cases}
\end{equation}

Following \citet{Young1993}, we call $P^\xi$ a regular perturbation of $P^0$ if $P^\xi$ is irreducible and aperiodic for all sufficiently small $\xi>0$, converges to $P^0$ as $\xi\to0$, and satisfies
\[
0<\lim_{\xi\to0}P^\xi_{zz'}\xi^{-r_{zz'}}<\infty
\]
for some $r_{zz'}\ge0$ for every positive transition. The exponent $r_{zz'}$ is called the resistance of the transition from $z$ to $z'$.

\begin{lemma}\label{lemma:blue-regular}
The family $P^\xi$ is a regular perturbation of $P^0$, with transition resistances
\[
r_{zz'}=\sum_{\substack{i\in I_c(z),\ b_i'\ne b_i}}f_i\bigl(V_i(\mu,\pi_i,b_i)\bigr).
\]
\end{lemma}

\begin{proof}
Irreducibility and aperiodicity follow from Lemma~\ref{lemma: P_ep}. We verify convergence to $P^0$ and the claimed resistance formula by first computing transition probabilities with error-free tests, and then bounding the effect of testing errors.

Fix $z=(b,\pi)$ and $z'=(b',\pi')$ in $\mathcal Z$. Introduce an auxiliary process $\widetilde P^\xi$ in which every performed test returns the correct outcome, while testing decisions, exploration, and resampling follow the original rules. Unlike $P^0$, this auxiliary process retains exploration at positive $\xi$. Let $\widetilde P^\xi_{i,zz'}$ denote player $i$'s probability of updating from $b_i$ to $b_i'$ in this process.

If $i\in I_c(z)$, a performed test fails to reject, so player $i$ resamples with probability $\gamma_i\xi^{f_i(V_i(\mu,\pi_i,b_i))}$. Accounting for both retaining the belief without resampling and drawing it again when resampling gives $\widetilde P^\xi_{i,zz'}=(1-\gamma_i\xi^{f_i(V_i(\mu,\pi_i,b_i))})\mathbf 1\{b_i'=b_i\}+\gamma_i\xi^{f_i(V_i(\mu,\pi_i,b_i))}\psi_i(b_i'\mid b_i)$. Thus, if $b_i'\ne b_i$, this probability equals $\gamma_i\psi_i(b_i'\mid b_i)\xi^{f_i(V_i(\mu,\pi_i,b_i))}$; if $b_i'=b_i$, it converges to one because the transformed value is strictly positive.

If $i\notin I_c(z)$, a performed test rejects, so player $i$ resamples whenever testing and otherwise retains their belief. Their update probabilities therefore coincide with those in \eqref{eq:blue-P0}: $\widetilde P^\xi_{i,zz'}=P^0_{i,zz'}>0$, independently of $\xi$.

The auxiliary test outcomes are determined by $z$, so independence of the update randomizations gives $\widetilde P^\xi_{zz'}=\prod_{i\in I}\widetilde P^\xi_{i,zz'}$. The factors associated with consistent players who change belief contribute exactly the powers of $\xi$ appearing in the stated resistance $r_{zz'}$. Dividing by $\xi^{r_{zz'}}$ and taking the limit therefore yields $\lim_{\xi\to0}\widetilde P^\xi_{zz'}\xi^{-r_{zz'}}=\bigl(\prod_{\{i\in I_c(z):b_i'\ne b_i\}}\gamma_i\psi_i(b_i'\mid b_i)\bigr)\bigl(\prod_{i\notin I_c(z)}P^0_{i,zz'}\bigr)$. Every factor is strictly positive and finite, with empty products equal to one. Hence this limit belongs to $(0,\infty)$. Moreover, the individual update probabilities converge to their counterparts in \eqref{eq:blue-P0}, establishing $\widetilde P^\xi\to P^0$.

We next compare the actual and auxiliary processes. Starting both at $z$, use the same within-epoch observations and the same randomizations for testing decisions, exploration, and resampling, but replace each performed test outcome by its correct outcome in the auxiliary process. If every performed test in the actual process is correct, the two processes produce identical next states.

Let $E$ denote the event that at least one performed test is incorrect. Assumption~\ref{ass:error_control} bounds each player's test error probability by $\xi^{\bar r}$. Since the decision to test is independent of the observations, $\Pr(E\mid z)\le\sum_{i\in I}\gamma_i\xi^{\bar r}\le |I|\xi^{\bar r}$. The coupled next states agree outside $E$, so $|P^\xi_{zz'}-\widetilde P^\xi_{zz'}|\le\Pr(E\mid z)\le |I|\xi^{\bar r}$. This bound uses a union bound and does not require independent test outcomes.

Since $\widetilde P^\xi\to P^0$ and the comparison bound tends to zero, we obtain $P^\xi\to P^0$. Furthermore, Assumption~\ref{ass:error_control} ensures $r_{zz'}\le\sum_{i\in I}f_i(V_i(\mu,\pi_i,b_i))<\bar r$. Consequently, $|P^\xi_{zz'}\xi^{-r_{zz'}}-\widetilde P^\xi_{zz'}\xi^{-r_{zz'}}|\le |I|\xi^{\bar r-r_{zz'}}\to0$. The normalized actual transition probability therefore converges to the same positive finite limit computed for the auxiliary process. This verifies the remaining regular-perturbation conditions and proves the resistance formula.
\end{proof}

\medskip\noindent\textit{2. Recurrent classes.} Once a Markov chain enters a recurrent communication class, it never leaves, and any state in the class can be reached from any other state in the class with positive probability.

\begin{lemma}\label{lemma:blue-recurrent}
The recurrent classes of $P^0$ are exactly the singletons $\{z\}$ with $z\in\mathcal Z^\dagger$.
\end{lemma}
\begin{proof}
At any consistent state $z\in\mathcal{Z}^\dagger$, every player retains their belief under $P^0$ by \eqref{eq:blue-P0}. Because policies are uniquely determined by beliefs, the process remains at $z$ with probability one. Thus each consistent state is absorbing and forms a singleton recurrent class.

To rule out other recurrent classes, we show that every inconsistent state can reach a consistent state with positive probability. Fix $z=(b,\pi)\notin\mathcal{Z}^\dagger$, and choose a player $i$ whose belief is inconsistent, so that $\max_{s\in S}\|b_i(s)-\pi_{-i}(s)\|_2>\tau$. We construct a path $z\to\tilde z\to w$, where every player is inconsistent at the intermediate state $\tilde z$, and $w$ is consistent.

By Assumption~\ref{assump: bad strategy}, there exists a belief $\tilde b_i\in \mathcal B_i$, with induced policy $\tilde\pi_i=\operatorname{Br}_i^\sigma(\tilde b_i)$, satisfying $\max_{s\in S}\|\tilde b_i(s)-\pi_{-i}(s)\|_2>\tau$ and $\max_{s\in S}\|\tilde\pi_i(s)-b_{ji}(s)\|_2>\tau$ for every $j\ne i$. Define the intermediate state by replacing only player $i$'s belief and policy: $\tilde b=(\tilde b_i,b_{-i})$, $\tilde\pi=(\tilde\pi_i,\pi_{-i})$, and $\tilde z=(\tilde b,\tilde\pi)$.

We verify that every player is inconsistent at $\tilde z$. Player $i$ is inconsistent by the first inequality, since their opponents' policies remain unchanged. For any other player $j$, their belief remains $b_j$, while player $i$'s policy changes to $\tilde\pi_i$. The discrepancy concerning player $i$ is one component of the discrepancy between player $j$'s belief and the opponents' policy profile. Hence, $\max_{s\in S}\|\tilde b_j(s)-\tilde\pi_{-j}(s)\|_2\ge\max_{s\in S}\|b_{ji}(s)-\tilde\pi_i(s)\|_2>\tau$. Thus player $j$ is also inconsistent. The game state witnessing inconsistency may differ across players.

We next verify that both transitions in the proposed path have positive probability under $P^0$. For the first transition, consider the event that player $i$ tests and resamples $\tilde b_i$, while every other player does not test. Since player $i$ is inconsistent at $z$, their test rejects with probability one in the unperturbed process. All other players retain their beliefs on this event. Therefore, $P^0_{z\tilde z}\ge\gamma_i\psi_i(\tilde b_i\mid b_i)\prod_{j\ne i}(1-\gamma_j)>0$.

For the second transition, choose any consistent state $w=(b^w,\pi^w)\in\mathcal{Z}^\dagger$, whose existence follows from Proposition~\ref{lemma: sigma, tau, M}. Every player is inconsistent at $\tilde z$, so every performed test rejects under $P^0$. If all players test and player $j$ draws $b_j^w$, the resulting belief profile is $b^w$, and the policy update produces $\pi^w$. Consequently, $P^0_{\tilde z w}\ge\prod_{j\in I}\gamma_j\psi_j(b_j^w\mid\tilde b_j)>0$. Both lower bounds are strictly positive because $\gamma_j\in(0,1)$ and the resampling distributions have full support. It follows that $[(P^0)^2]_{zw}\ge P^0_{z\tilde z}P^0_{\tilde z w}>0$.

Therefore, every inconsistent state can reach an absorbing consistent state in at most two transitions. If an inconsistent state belonged to a recurrent class, closure would require that class to contain the reachable consistent state. This is impossible because the consistent state is absorbing and cannot return to the inconsistent state. Hence there are no recurrent classes other than the singleton consistent states.
\end{proof}

\medskip\noindent\textit{3. Resistance and stochastic potential.}
The preceding results establish that $P^\xi$ is a regular perturbation of $P^0$ and that the recurrent classes of $P^0$ are exactly the singleton consistent states. Although all consistent states are absorbing under $P^0$, perturbations allow transitions between them. We use the resistance-tree framework of \citet{Young1993} to determine which of these states are stochastically stable.

Recall that the transition resistance $r_{zz'}$ is the leading-order exponent of $P^\xi_{zz'}$ as $\xi\to0$. For a finite path $\rho=(z_0\to z_1\to\cdots\to z_\ell)$, its resistance is the sum of the transition resistances, $R(\rho):=\sum_{k=0}^{\ell-1}r_{z_kz_{k+1}}$. For distinct consistent states $z,w\in\mathcal Z^\dagger$, define $\widehat r_{zw}:=\min_{\rho:z\to w}R(\rho)$, where the minimum is over all finite paths in $\mathcal Z$ from $z$ to $w$. Thus, $\widehat r_{zw}$ measures the least total resistance required to move from one consistent state to another, possibly through intermediate inconsistent states. We show that the minimum path resistance $\widehat r_{zw}$ depends only on the departing state $z$.

\begin{lemma}\label{lemma:blue-path-resistance}
For distinct $z=(b,\pi),w\in\mathcal Z^\dagger$,
$\widehat r_{zw}=\min_{i\in I}f_i\bigl(V_i(\mu,\pi_i,b_i)\bigr)$.
\end{lemma}

\begin{proof}
Fix distinct states $z=(b,\pi),w\in\mathcal Z^\dagger$. We establish a lower bound for the resistance of every path from $z$ to $w$ and then construct a path that attains this bound.

First, consider any finite path $\rho$ from $z$ to $w$, and let $z'=(b',\pi')$ be its first state different from $z$. Such a state exists because $w\ne z$. Since policies are uniquely determined by beliefs, $z'\ne z$ implies that $b_j'\ne b_j$ for at least one player $j$. Every player is consistent at $z$, so Lemma~\ref{lemma:blue-regular} gives $r_{zz'}=\sum_{\{j\in I:b_j'\ne b_j\}}f_j(V_j(\mu,\pi_j,b_j))\ge\min_{j\in I}f_j(V_j(\mu,\pi_j,b_j))$. All transition resistances are nonnegative, and therefore $R(\rho)\ge r_{zz'}\ge\min_{j\in I}f_j(V_j(\mu,\pi_j,b_j))$. Since this holds for every path from $z$ to $w$, we obtain $\widehat r_{zw}\ge\min_{j\in I}f_j(V_j(\mu,\pi_j,b_j))$.

To establish the reverse inequality, choose a player $i\in\arg\min_{j\in I}f_j(V_j(\mu,\pi_j,b_j))$. Assumption~\ref{assump: bad strategy} provides a belief $\tilde b_i\in \mathcal B_i$, with induced policy $\tilde\pi_i=\operatorname{Br}_i^\sigma(\tilde b_i)$, such that $\max_{s\in S}\|\tilde b_i(s)-\pi_{-i}(s)\|_2>\tau$ and $\max_{s\in S}\|\tilde\pi_i(s)-b_{ji}(s)\|_2>\tau$ for every $j\ne i$. Define $\tilde b=(\tilde b_i,b_{-i})$, $\tilde\pi=(\tilde\pi_i,\pi_{-i})$, and $\tilde z=(\tilde b,\tilde\pi)\in\mathcal Z$. These inequalities ensure that every player is inconsistent at $\tilde z$. Player $i$ is inconsistent by the first inequality, since their opponents' policies remain unchanged. For any $j\ne i$, their belief remains $b_j$, but the policy of opponent $i$ changes to $\tilde\pi_i$. Consequently, $\max_{s\in S}\|\tilde b_j(s)-\tilde\pi_{-j}(s)\|_2\ge\max_{s\in S}\|b_{ji}(s)-\tilde\pi_i(s)\|_2>\tau$, so player $j$ is also inconsistent. Moreover, $\tilde b_i\ne b_i$: otherwise, player $i$'s belief and the opponents' policies both are unchanged, contradicting consistency at $z$ and inconsistency at $\tilde z$.

Now consider the path $z\to\tilde z\to w$. In the first transition, player $i$ changes their belief to $\tilde b_i$, while every other player retains their belief. Since all players are consistent at $z$ and only player $i$ changes belief, Lemma~\ref{lemma:blue-regular} gives $r_{z\tilde z}=f_i(V_i(\mu,\pi_i,b_i))=\min_{j\in I}f_j(V_j(\mu,\pi_j,b_j))$.

In the second transition, players resample the beliefs specified by $w$. Every player is inconsistent at $\tilde z$, so under $P^0$ each player can test, reject, and draw the required belief with positive probability. Thus this transition requires no exploration by a consistent player, and Lemma~\ref{lemma:blue-regular} gives $r_{\tilde z w}=0$. The constructed path therefore has resistance $R(z\to\tilde z\to w)=r_{z\tilde z}+r_{\tilde z w}=\min_{j\in I}f_j(V_j(\mu,\pi_j,b_j))$. By the definition of minimum path resistance, $\widehat r_{zw}\le\min_{j\in I}f_j(V_j(\mu,\pi_j,b_j))$. Combining this with the lower bound proves the claimed equality.
\end{proof}

We now use these path resistances to identify the stochastically stable states. By Lemma~\ref{lemma:blue-recurrent}, the recurrent classes of $P^0$ are the singletons $\{z\}$ with $z\in\mathcal Z^\dagger$. Consider the directed graph on $\mathcal Z^\dagger$ in which the edge $(w,w')$ from $w$ to $w'\neq w$ has weight $\widehat r_{ww'}$. For $z\in\mathcal Z^\dagger$, a $z$-tree is a set of edges in which every $w\in\mathcal Z^\dagger\setminus\{z\}$ has exactly one outgoing edge, $z$ has none, and there is no directed cycle, so that following the edges from any $w\neq z$ leads to $z$. Let $\mathrm{Tree}_z$ denote the set of all $z$-trees, which is nonempty since every $w\neq z$ may point directly to $z$. The \emph{stochastic potential} of $z$ is $\Phi(z):=\min_{\mathbb T\in\mathrm{Tree}_z}\sum_{(w,w')\in\mathbb T}\widehat r_{ww'}$.

\begin{lemma}[\citet{Young1993}]\label{theorem: Young 1993}
For a regular perturbation of a finite Markov chain, the stochastically stable states are precisely those belonging to recurrent classes with minimum stochastic potential.
\end{lemma}

\emph{Proof of Theorem \ref{theorem:main}.} For a fixed root $z=(b,\pi)$, every $z$-tree contains exactly one outgoing edge from each vertex $w=(b^w,\pi^w)\ne z$, and none from $z$. By Lemma~\ref{lemma:blue-path-resistance}, the resistance of the edge leaving $w$ depends only on $w$. Thus every $z$-tree sums the same contributions, one from each $w\ne z$, regardless of how the vertices are connected. Its total resistance therefore equals $\Phi(z)$, giving
\[
\Phi(z) =\sum_{w\in\mathcal Z^\dagger\setminus\{z\}}\min_{i\in I}f_i\bigl(V_i(\mu,\pi_i^w,b_i^w)\bigr) =\sum_{w\in\mathcal Z^\dagger}\min_{i\in I}f_i\bigl(V_i(\mu,\pi_i^w,b_i^w)\bigr)-\min_{i\in I}f_i\bigl(V_i(\mu,\pi_i,b_i)\bigr).
\]
The sum in the last expression is independent of $z$. Lemma~\ref{theorem: Young 1993} therefore implies that $z=(b,\pi)$ is stochastically stable if and only if $z\in\mathcal Z^\dagger$ and
\[
\min_{i\in I}f_i\bigl(V_i(\mu,\pi_i,b_i)\bigr)=\max_{z'=(b',\pi')\in\mathcal Z^\dagger}\min_{i\in I}f_i\bigl(V_i(\mu,\pi_i',b_i')\bigr).
\]

This proves the selection rule in Theorem~\ref{theorem:main}. The approximate-equilibrium conclusion follows from Proposition~\ref{lemma: sigma, tau, M}. When $f_i(x)=x$ for all $i\in I$, replacing $f_i(V_i)$ with $V_i$ throughout gives
$\min_{i\in I}V_i(\mu,\pi_i,b_i)=\max_{z'=(b',\pi')\in\mathcal Z^\dagger}\min_{i\in I}V_i(\mu,\pi_i',b_i')$. \hfill $\square$

\bibliographystyle{IEEEtranN}
\bibliography{library}

\begin{thebibliography}{21}
\providecommand{\natexlab}[1]{#1}
\providecommand{\url}[1]{#1}
\csname url@samestyle\endcsname
\providecommand{\newblock}{\relax}
\providecommand{\bibinfo}[2]{#2}
\providecommand{\BIBentrySTDinterwordspacing}{\spaceskip=0pt\relax}
\providecommand{\BIBentryALTinterwordstretchfactor}{4}
\providecommand{\BIBentryALTinterwordspacing}{\spaceskip=\fontdimen2\font plus
\BIBentryALTinterwordstretchfactor\fontdimen3\font minus
  \fontdimen4\font\relax}
\providecommand{\BIBforeignlanguage}[2]{{%
\expandafter\ifx\csname l@#1\endcsname\relax
\typeout{** WARNING: IEEEtranN.bst: No hyphenation pattern has been}%
\typeout{** loaded for the language `#1'. Using the pattern for}%
\typeout{** the default language instead.}%
\else
\language=\csname l@#1\endcsname
\fi
#2}}
\providecommand{\BIBdecl}{\relax}
\BIBdecl

\bibitem[Young(1993)]{Young1993}
H.~P. Young, ``The evolution of conventions,'' \emph{Econometrica}, vol.~61,
  no.~1, pp. 57--84, 1 1993.

\bibitem[Marden and Shamma(2012)]{MardenShamma2012}
J.~R. Marden and J.~S. Shamma, ``Revisiting log-linear learning: Asynchrony,
  completeness and payoff-based implementation,'' \emph{Games and Economic
  Behavior}, vol.~75, pp. 788--808, 7 2012.

\bibitem[Pradelski and Young(2012)]{Pradelski2012}
B.~S. Pradelski and H.~P. Young, ``Learning efficient nash equilibria in
  distributed systems,'' \emph{Games and Economic Behavior}, vol.~75, pp.
  882--897, 7 2012.

\bibitem[Marden et~al.(2014)Marden, Young, and Pao]{MardenYoungPao2012}
J.~R. Marden, H.~P. Young, and L.~Y. Pao, ``Achieving pareto optimality through
  distributed learning,'' \emph{SIAM Journal on Control and Optimization},
  vol.~52, no.~5, pp. 2753--2770, 2014.

\bibitem[Yongacoglu et~al.(2022)Yongacoglu, Arslan, and
  Y{\"u}ksel]{yongacoglu2022decentralized}
B.~Yongacoglu, G.~Arslan, and S.~Y{\"u}ksel, ``Decentralized learning for
  optimality in stochastic dynamic teams and games with local control and
  global state information,'' \emph{IEEE Transactions on Automatic Control},
  vol.~67, no.~10, pp. 5230--5245, 2022.

\bibitem[Donmez et~al.(2025)Donmez, Unlu, and Sayin]{donmez2025logitq}
A.~S. Donmez, O.~Unlu, and M.~O. Sayin, ``Logit-{Q} dynamics for efficient
  learning in stochastic teams,'' \emph{SIAM Journal on Control and
  Optimization}, vol.~63, no.~5, pp. 3219--3243, 2025.

\bibitem[Christianos et~al.(2023)Christianos, Papoudakis, and
  Albrecht]{christianos2023pareto}
F.~Christianos, G.~Papoudakis, and S.~V. Albrecht, ``Pareto actor-critic for
  equilibrium selection in multi-agent reinforcement learning,''
  \emph{Transactions on Machine Learning Research}, 2023.

\bibitem[Zhang et~al.(2024)Zhang, Zardini, Ozdaglar, Shamma, and
  Li]{zhang2024equilibrium}
R.~Zhang, G.~Zardini, A.~Ozdaglar, J.~Shamma, and N.~Li, ``Equilibrium
  selection for multi-agent reinforcement learning: A unified framework,''
  2024, arXiv:2406.08844, revised January 2026.

\bibitem[Kiremitci et~al.(2026)Kiremitci, Donmez, and
  Sayin]{kiremitci2026decentralized}
S.~T. Kiremitci, A.~S. Donmez, and M.~O. Sayin, ``Decentralized optimal
  equilibrium learning in stochastic games via single-bit feedback,'' 2026,
  arXiv:2602.12830.

\bibitem[Foster and Young(2003)]{FosterYoung2003}
D.~P. Foster and H.~P. Young, ``Learning, hypothesis testing, and {Nash}
  equilibrium,'' \emph{Games and Economic Behavior}, vol.~45, no.~1, pp.
  73--96, 2003.

\bibitem[Jindani(2022)]{Jindani2022}
S.~Jindani, ``Learning efficient equilibria in repeated games,'' \emph{Journal
  of Economic Theory}, vol. 205, p. 105551, 2022.

\bibitem[Calvano et~al.(2020)Calvano, Calzolari, Denicolò, and
  Pastorello]{Calvano}
E.~Calvano, G.~Calzolari, V.~Denicolò, and S.~Pastorello, ``Artificial
  intelligence, algorithmic pricing, and collusion,'' \emph{American Economic
  Review}, vol. 110, no.~10, pp. 3267--3297, October 2020.

\bibitem[Klein(2021)]{klein2021autonomous}
T.~Klein, ``Autonomous algorithmic collusion: {Q}-learning under sequential
  pricing,'' \emph{The RAND Journal of Economics}, vol.~52, no.~3, pp.
  538--558, 2021.

\bibitem[Calvano et~al.(2021)Calvano, Calzolari, Denicolò, and
  Pastorello]{calvano2021imperfect}
E.~Calvano, G.~Calzolari, V.~Denicolò, and S.~Pastorello, ``Algorithmic
  collusion with imperfect monitoring,'' \emph{International Journal of
  Industrial Organization}, vol.~79, p. 102712, 2021.

\bibitem[Asker et~al.(2024)Asker, Fershtman, and Pakes]{asker2024impact}
J.~Asker, C.~Fershtman, and A.~Pakes, ``The impact of artificial intelligence
  design on pricing,'' \emph{Journal of Economics \& Management Strategy},
  vol.~33, no.~2, pp. 276--304, 2024.

\bibitem[Mailath and Olszewski(2011)]{MailathOlszewski2011}
G.~J. Mailath and W.~Olszewski, ``Folk theorems with bounded recall under
  (almost) perfect monitoring,'' \emph{Games and Economic Behavior}, vol.~71,
  pp. 174--192, 2011, special Issue In Honor of John Nash.

\bibitem[Cartea et~al.(2026)Cartea, Chang, Penalva, and
  Waldon]{cartea2026algorithmic}
{\'A}.~Cartea, P.~Chang, J.~Penalva, and H.~Waldon, ``Algorithmic collusion and
  a folk theorem from learning with bounded rationality,'' \emph{Games and
  Economic Behavior}, vol. 157, pp. 1--21, 2026.

\bibitem[Askenazi-Golan et~al.(2026)Askenazi-Golan, Mergoni~Cecchelli, Plumb,
  and Possnig]{AskenaziGolanEtAl2026}
\BIBentryALTinterwordspacing
G.~Askenazi-Golan, D.~Mergoni~Cecchelli, E.~Plumb, and C.~Possnig, ``The bounds
  of algorithmic collusion: {Q}-learning, gradient learning, and the folk
  theorem,'' 2026, version 2. [Online]. Available:
  \url{https://arxiv.org/abs/2411.12725v2}
\BIBentrySTDinterwordspacing

\bibitem[Yang and Wu(2025)]{ruifan}
R.~Yang and M.~Wu, ``Learning with episodic hypothesis testing in general
  games: A framework for equilibrium selection,'' \emph{Proceedings of the ACM
  on Measurement and Analysis of Computing Systems}, vol.~9, no.~3, Dec. 2025.

\bibitem[Geist et~al.(2019)Geist, Scherrer, and Pietquin]{geist2019theory}
M.~Geist, B.~Scherrer, and O.~Pietquin, ``A theory of regularized markov
  decision processes,'' in \emph{International conference on machine
  learning}.\hskip 1em plus 0.5em minus 0.4em\relax PMLR, 2019, pp. 2160--2169.

\bibitem[Maheshwari et~al.(2025)Maheshwari, Wu, Pai, and
  Sastry]{maheshwari2025independent}
C.~Maheshwari, M.~Wu, D.~Pai, and S.~Sastry, ``Independent and decentralized
  learning in markov potential games,'' \emph{IEEE Transactions on Automatic
  Control}, vol.~70, no.~11, pp. 7538--7553, 2025.

\end{thebibliography}

\end{document}